\documentclass[a4paper,11pt]{article}
\usepackage{graphicx} 
\usepackage{amsmath}
\usepackage{amssymb}
\usepackage{amsfonts}
\usepackage{amsthm}
\usepackage{parskip}
\usepackage{comment}
\usepackage{xcolor}
\usepackage{tikz}
\usetikzlibrary{calc}
\usepackage{yhmath}
\usepackage{enumitem}
\usepackage[top=22mm, bottom=23mm, left=22mm, right=22mm]{geometry}
\allowdisplaybreaks

\newtheorem{thm}{Theorem}[section]
\newtheorem{prop}[thm]{Proposition}
\newtheorem{obs}[thm]{Observation}
\newtheorem{lem}[thm]{Lemma}
\newtheorem{cor}[thm]{Corollary}

\newtheorem{ques}[thm]{Question}

\theoremstyle{definition}

\newtheorem{df}[thm]{\bf Definition}
\newtheorem{construction}[thm]{\bf Construction}

\title{Combinatorial foundations for solvable chaotic local Euclidean quantum circuits in two dimensions}
\author{Fredy Yip\thanks{Trinity College, University of Cambridge, United Kingdom. Email: \textbf{fy276@cam.ac.uk}.}}
\date{}

\begin{document}

\maketitle

\begin{abstract}
    We investigate a graph-theoretic problem motivated by questions in quantum computing concerning the propagation of information in quantum circuits. A graph $G$ is said to be a \emph{bounded extension} of its subgraph $L$ if they share the same vertex set, and the graph distance $d_L(u, v)$ is uniformly bounded for edges $uv\in G$. Given vertices $u, v$ in $G$ and an integer $k$, the \emph{geodesic slice} $S(u, v, k)$ denotes the subset of vertices $w$ lying on a geodesic in $G$ between $u$ and $v$ with $d_G(u, w) = k$. We say that $G$ has \emph{bounded geodesic slices} if $|S(u, v, k)|$ is uniformly bounded over all $u, v, k$. We call a graph $L$ \emph{geodesically directable} if it has a bounded extension $G$ with bounded geodesic slices. 

    Contrary to previous expectations, we prove that $\mathbb{Z}^2$ is geodesically directable. Physically, this provides a setting in which one could devise exactly-solvable chaotic local quantum circuits with non-trivial correlation patterns on 2D Euclidean lattices. In fact, we show that any bounded extension of $\mathbb{Z}^2$ is geodesically directable. This further implies that all two-dimensional regular tilings are geodesically directable. 
\end{abstract}

\section{Introduction}

The study of quantum circuits offers a pathway to understanding general quantum many-body systems~\cite{review}. A quantum circuit is formed from qubits residing on sites of a lattice with adjacent site interactions. Quantum gates between pairs of sites may be added to the lattice, allowing interaction between the pair. These quantum gates are often physically constrained to be \emph{local}, only being able to connect pairs of sites already close in the original lattice. The system evolves in discrete time steps according to these interactions between adjacent sites and sites spanned by quantum gates. Discrete time evolution of this form has served as a structured setting to gain insight on fundamental physics questions involving the spreading of quantum information and thermalisation, whilst also being naturally realised in modern digital quantum platforms such as superconducting qubit arrays~\cite{arutesuperconducting, kimsuperconducting}. 

Mathematically, we may model sites as vertices and interactions as edges in an unweighted graph. We model the initial lattice as a connected graph $L$. We add (potentially infinitely many) edges to $L$ to form the graph $G$ which accounts for interactions both between adjacent sites and added quantum gates. $L$ and $G$ share the same vertex set. The graphs we consider here are unweighted, however, weighted graphs will later be introduced in the proof. 

The notion of closeness of vertices in a graph is captured by the \emph{graph distance}. In a connected unweighted graph, the length of a path is the number of edges it uses. The distance between two vertices is the length of a shortest path between them, such shortest paths are known as \emph{geodesics}. This distance depends on the graph we are working in. For example, whilst $L$ and $G$ share the same set of vertices, the distance between the same pair of vertices may be strictly smaller in $G$ than in $L$. 

The locality requirement on the quantum gates translates to the uniform boundedness of the distance (in the original lattice $L$) between endpoints of an added edge in $G$. That is, there exists a positive integer $M$ such that a path of length at most $M$ exists in $L$ between the two endpoints of any new edge in $G$. 

\begin{df}
    For graphs $X, Y$ on the same vertex set, we say that $X$ is a \emph{bounded extension} of $Y$ if $Y$ is a subgraph of $X$ and the distance in $Y$ between $u$ and $v$ is uniformly bounded over edges $uv\in X$. 
\end{df}

The addition of local quantum gates to the base lattice $L$ allows us to form any bounded extension $G$ of $L$. 

Physically, the graph distance on $G$ constrains the propagation of quantum information within a given number of discrete time steps. An important probe of the dynamics on the system is two-point correlation functions. Great interest has been taken in certain classes of solvable quantum circuits, in which one can entirely understand the behaviour of two-point correlation functions, providing rare instances of exactly-solvable chaotic quantum dynamics~\cite{dualityreview}. For these circuits, the only non-vanishing correlation functions are those where the time separation matches the graph distance between the endpoints $u, v$. Here, the relevant interactions propagate through vertices lying on geodesics between $u, v$. 

\begin{df}
    For vertices $u, v$ in a graph $X$ and an integer $k$, let the \emph{geodesic slice} $S(u, v, k) = S_X(u, v, k)$ denote the set of vertices lying on a geodesic between $u, v$ and having distance $k$ from $u$ in $X$. 
\end{df}

Alternatively, $S_X(u, v, k)$ is the set of vertices with distance $k$ to $u$ and $d - k$ to $v$ in $X$, where $d$ is the distance between $u$ and $v$ in $X$. Note that if $k < 0$ or $k > d$, $S_X(u, v, k)$ is empty. 

As pointed out in Breach, Placke, Claeys, and Parameswaran~\cite{breach}, efficient computation of correlation functions requires the geodesic slices to have bounded sizes. 

\begin{df}
    A graph $X$ has \emph{bounded geodesic slices} if the size of the geodesic slice $S_X(u, v, k)$ is uniformly bounded across all choices of vertices $u, v$ and integers $k$. 
\end{df}

Breach, Placke, Claeys, and Parameswaran~\cite{breach} noted that if $G$ has bounded geodesic slices, then the dynamics of the corresponding system can be understood through quantum channels of linear depth. This condition is weaker than the boundedness of the number of geodesics in $G$ between any pair of vertices $u, v$, as each such geodesic contains a single vertex in the slice $S_G(u, v, k)$ for each choice of $k$. 

Consequentially, the natural question is whether a given base lattice $L$ admits a bounded extension $G$ (via the addition of local quantum gates) which has bounded geodesic slices. 

\begin{df}
    We call a graph $L$ \emph{geodesically directable}, if there exists a bounded extension $G$ of $L$ with bounded geodesic slices. 
\end{df}

In proving the geodesic directability of a graph $L$, two opposing considerations must be balanced as we choose the bounded extension $G$. On one hand, in passing from $L$ to $G$, the metric structure is minimally altered, as $G$ and $L$ are quasi-isometric. On the other hand, if $L$ does not already have bounded geodesic slices, the requirement that $G$ has bounded geodesic slices necessitates a drastic alteration of the geodesic structure. Loosely speaking, the geodesics in $G$ must become much more unique than those of the base lattice $L$.  

\begin{ques} \label{general ques}
    Which graphs $L$ are geodesically directable? 
\end{ques}

An affirmative answer to Question \ref{general ques} could be given trivially when $L$ is a finite graph. In this case, $L$ itself has bounded geodesic slices, as the size of a geodesic slice is bounded above by the size of $L$. 

Breach, Placke, Claeys, and Parameswaran~\cite{breach} gave an affirmative answer to Question \ref{general ques} in the case where $L$ is (Gromov) hyperbolic and of bounded degree. In this case, $L$ itself again has bounded geodesic slices. To see this, let $L$ be $\delta$-hyperbolic in the sense of Rips (see \cite{Bridson-Haefliger}), for some $\delta\geq 0$. Consider any pair of vertices $w_1, w_2\in S_L(u, v, k)$ lying on geodesics $\gamma_1, \gamma_2$ between $u, v$, respectively. Applying the $\delta$-hyperbolicity of $L$ to the degenerate geodesic triangle formed by $\gamma_1, \gamma_2$ implies that $w_1$ has distance at most $\delta$ from some vertex $w'$ on $\gamma_2$. By the triangle inequality, the distance from $w'$ to $u$ is between $k - \delta$ and $k + \delta$. As a result, the distance between $w'$ and $w_2$ is at most $\delta$. Again, by the triangle inequality, the distance between $w_1$ and $w_2$ is at most $2\delta$. Therefore, $S_L(u, v, k)$ has diameter at most $2\delta$. As $L$ has bounded degree, the size of any such subset of vertices is uniformly bounded. 

The argument above remains valid whenever $L$ admits a hyperbolic bounded extension $G$. However, since $L$ and $G$ are quasi-isometric as metric spaces, $G$ is hyperbolic if and only if $L$ is. Thus, the additional generality is vacuous. 

A particularly interesting case of Question~\ref{general ques} is when $L$ is a regular tiling. A regular tiling by $p$-gons with $q$ faces meeting at each vertex is denoted by its Schläfli symbol $\{p, q\}$ (\emph{cf.}~\cite{coxeter_regular_1973}). Physically, such tilings, being translationally and rotationally symmetric, form natural models of base lattices. 

\begin{ques} \label{lattice ques}
    Which regular tilings $\{p, q\}$ are geodesically directable? 
\end{ques}

As spherical tilings have finitely many vertices, and hyperbolic tilings are regular hyperbolic graphs, affirmative answers to Question \ref{lattice ques} hold for both families, leaving only Euclidean tilings. Here, there are only three cases, the square lattice $\{4, 4\}$ (\emph{i.e.}~$\mathbb{Z}^2$), the triangular lattice $\{3, 6\}$ and the hexagonal lattice $\{6, 3\}$. 

Euclidean tilings are the most relevant for physical applications. Spherical tilings (the five Platonic solids) yield small, finite systems that cannot be scaled. Likewise, only small patches of hyperbolic tilings can be embedded in three-dimensional Euclidean space, due to the exponential growth of the number of vertices as a function of radius. On the other hand, Euclidean tilings arise naturally as two-dimensional lattices in scalable physical systems. Implementing such dynamics on quantum devices, and comparing to the analytical solutions, would allow for accurate benchmarking of the quality of the many-body gates in the system. Furthermore, such structures could find interesting application in quantum error correction, since they share features with non-Euclidean lattices that host good quantum error correction codes \cite{pastawskiHaPPY, leverrierqLDPC, hastingsqLDPC}. 

Contrary to previous expectations \cite{breach}, we shall give affirmative answers to all three Euclidean lattices. Combined with the spherical and hyperbolic cases, this yields a fully affirmative answer to Question~\ref{lattice ques}. 

\begin{thm} \label{general thm}
    All (two-dimensional) regular tilings are geodesically directable. 
\end{thm}

To prove Theorem \ref{general thm}, we first address the case of the square lattice $\{4, 4\} = \mathbb{Z}^2$. 

\begin{thm} \label{z2 thm}
    The square lattice $\mathbb{Z}^2$ is geodesically directable. 
\end{thm}

The geodesic slices $S_{\mathbb{Z}^2}(u, v, k)$ of $\mathbb{Z}^2$ itself may have size as large as linear in the distance $d_{\mathbb{Z}^2}(u, v)$. For example, for any non-negative integer $x$, $|S_{\mathbb{Z}^2}((0, 0), (x, x), x)| = x + 1$. 

One way of understanding the origin of these linear-sized geodesic slices is to consider the possible displacements $u - v$ of an edge $uv$ in $\mathbb{Z}^2$. Specifically, any path in $\mathbb{Z}^2$ from $(0, 0)$ to $(x, x)$ consisting of $x$ steps of displacement $(0, 1)$ and $x$ steps of displacement $(1, 0)$, in any order, forms a geodesic. This phenomenon is in stark contrast with the unique geodesics of $\mathbb{R}^2$ equipped with the Euclidean $\ell_2$ norm, where the unit ball is smooth. 

It is possible to construct a bounded extension $G$ of $\mathbb{Z}^2$ inducing a metric on $\mathbb{Z}^2$ which approximates the Euclidean metric within any prescribed multiplicative error. However, in any fixed bounded extension $G$ of $\mathbb{Z}^2$, the number of possible displacements $u - v$ of edges $uv$ remains finite. For this reason, one may expect linear sized geodesic slices to emerge again in $G$ at sufficient large scales. 

Breach, Placke, Claeys, and Parameswaran~\cite{breach} observed that this is indeed the case for \emph{translation-invariant} bounded extensions of $\mathbb{Z}^2$. Here, we call a bounded extension $G$ of $\mathbb{Z}^2$ \emph{translation-invariant} if the adjacency of vertices $u, v$ in $G$ depends only on the displacement $u - v$ (\emph{i.e.}~$G$ is a finitely-generated Cayley graph). 

Therefore, to construct bounded extensions of $\mathbb{Z}^2$ with bounded geodesic slices, we must go beyond translation-invariant bounded extensions. 

On the other hand, there are simple constructions of bounded extensions $G$ of $\mathbb{Z}^2$ in which, for some fixed vertex $u$, the sizes of the geodesic slices $S(u, v, k)$ are bounded as $v$ and $k$ vary. To see this in the case where $u = (0, 0)$ is the origin, we may form $G$ from $\mathbb{Z}^2$ by adding all edges of the form $(x, 0)(x + 2, 0)$. In fact, it may further be shown that in this construction of $G$, the size of the geodesic slice $S_G(u, v, k)$ is uniformly bounded across all values of $k$ and all pairs of vertices $u, v$ lying on opposite sides of the $x$-axis. 

The two constructions of bounded extensions $G$ we sketched above suffer from distinct problems. On one hand, a bounded extension $G$ approximating the Euclidean metric leads to unbounded geodesic slices at large scales, \emph{i.e.}, when the separation $d(u, v)$ is sufficiently large. On the other hand, the bounded extension $G$ modifying the $x$-axis has unbounded geodesic slices when $u, v$ are localised on the same side of the $x$-axis, which may loosely be understood to be an obstruction at small scales. Therefore, to construct a bounded extension $G$ of $\mathbb{Z}^2$ with bounded geodesic slices, we need to account for the behaviour of geodesic slices at all scales. Motivated by these considerations, our construction of $G$ shall be fractal in nature, enjoying a form of self-similarity under scaling. 

We shall, in fact, prove a strengthening of Theorem~\ref{z2 thm}, that all bounded extensions of $\mathbb{Z}^2$ are geodesically directable. 

\begin{thm} \label{z2 gen ext thm}
    Any bounded extension of the square lattice $\mathbb{Z}^2$ is geodesically directable. 
\end{thm}

In physical terms, Theorem~\ref{z2 gen ext thm} states that after the addition of any set of local quantum gates to $\mathbb{Z}^2$, we may add a further set of local (but potentially longer) quantum gates to achieve a system with bounded geodesic slices. 

The two remaining cases of Theorem \ref{general thm}, namely the triangular lattice $\{3, 6\}$ and the hexagonal lattice $\{6, 3\}$, may be treated in a similar manner as the square lattice. We shall resolve both cases by reducing them to the square lattice and applying Theorem \ref{z2 gen ext thm}. 

In favour of a clear exposition, we shall not attempt to optimise the constants in our arguments. In fact, our results hold with much sharper constants in many cases. 

\textbf{Organisation of the Paper.} The rest of this paper is dedicated to the proof of Theorem~\ref{z2 gen ext thm} for bounded extensions of the square lattice $\mathbb{Z}^2$. Section~\ref{intro and weight} sets up the notation and reduces the problem to the existence of certain weighted graphs with bounded geodesic slices. Section~\ref{construction sec} gives the construction of such a weighted graph as a member of a family of weighted graphs. Section~\ref{fractal} explores a fractal-like property of this family of weighted graphs under scaling. Section~\ref{proof} concludes the proof of Theorem~\ref{z2 gen ext thm} for bounded extensions of the square lattice $\mathbb{Z}^2$. Section~\ref{other Euc} addresses the hexagonal and triangular lattices, concluding the proof of Theorem~\ref{general thm}. 

For the sake of practical applications in quantum circuits, we include an explicit construction of a bounded extension of $\mathbb{Z}^2$ with bounded geodesic slices in the appendix. 

\section{Reduction to weighted graphs} \label{intro and weight}

\begin{df}
    For an unweighted graph $G$, we say that \emph{$H$ is a weighted graph on $G$} if $H$ is obtained from $G$ by assigning weights to the edges of $G$. 
\end{df}

In this section, we show that Theorem \ref{z2 gen ext thm} follows from the existence of a certain weighted graph on $\mathbb{Z}^2$ with bounded geodesic slices. 

Throughout this paper, we consider only weighted graphs with positive weights. In fact, we will only consider weighted graphs with positive integer weights. All graphs considered will be connected. 

In a weighted graph, the \emph{cost} of a path is the sum of the weights of its edges. The weighted distance between two vertices is the infimum of the cost of paths between them. Here, in a connected graph with positive integer weights, this infimum is attained. Paths achieving this infimum are known as weighted geodesics. Unless otherwise stated, in a weighted graph, distances between vertices refer to weighted distances, and geodesics refer to weighted geodesics. All paths and geodesics considered here will be finite. In contrast to the weighted cost of a path, we use the \emph{length} of a path to denote the unweighted notion of the number of edges in the path. 

In a (weighted or unweighted) graph $X$, we denote the distance between vertices $u, v$ in $X$ by $d_X(u, v)$. We use $B(v, r) = B_X(v, r) = \{u\in X|d_X(u, v)\leq r\}$ to denote the closed ball of radius $r$ centred at the vertex $v$. Throughout this paper, we shall use the crude estimate
\begin{equation*}
    |B_{\mathbb{Z}^2}(v, r)| = 2\lfloor r\rfloor^2 + 2\lfloor r\rfloor + 1\leq 5r^2, 
\end{equation*}
for $r\geq 1$. 

\begin{df}
    For a real number $r$ and a pair of vertices $u, v$ in a weighted graph $X$, the \emph{geodesic slice} $S_X(u, v, r)$ is defined to be the set of vertices lying on a (weighted) geodesic between $u, v$ and having (weighted) distance $r$ to $u$. 
\end{df}

\begin{df}
    A weighted graph $X$ is said to have \emph{bounded geodesic slices} if the size of geodesic slices $S_X(u, v, r)$ is bounded across all choices of endpoints $u, v$ and non-negative real numbers $r$. 
\end{df}

We are now ready to state the main proposition of this section, which reduces Theorem~\ref{z2 gen ext thm} to the existence of a weighted graph $H$ on $\mathbb{Z}^2$ with finitely many positive integer weights and bounded geodesic slices. 

\begin{prop} \label{sec2 reduction prop}
    If there exists a weighted graph $H$ on $\mathbb{Z}^2$ using finitely many positive integer weights with bounded geodesic slices, then any bounded extension $L$ of $\mathbb{Z}^2$ is geodesically directable. 
\end{prop}

The rest of this section is dedicated to the proof of Proposition~\ref{sec2 reduction prop}. We shall first introduce the relevant notation and intermediate results, some of which will remain applicable to later sections. 

As $L$ is a bounded extension of $\mathbb{Z}^2$, there exists a constant $M$ such that for every edge $uv$ of $L$, $d_{\mathbb{Z}^2}(u, v)\leq M$. 

\begin{obs} \label{L vs Z^2}
    For any pair of vertices $u, v\in \mathbb{Z}^2$, we have
    \begin{equation*}
        d_L(u, v)\leq d_{\mathbb{Z}^2}(u, v)\leq Md_L(u, v). 
    \end{equation*}
\end{obs}

Let $W$ be the maximum weight of $H$. To form a bounded extension $G$ of $L$ with bounded geodesic slices, we shall add horizontal and vertical edges to $L$ to emulate $H$ on $(N\mathbb{Z})^2$ for some appropriate choice of a positive integer $N$. Formally, we construct $G$ as follows. 
    
\begin{construction} \label{reduction pattern}
    Let $N > 10WM^2$ be any common multiple of the weights of $H$. We form $G$ by adding the following edges to $L$. 
    \begin{itemize}
        \item If a horizontal edge $(x, y)(x + 1, y)$ has weight $w$ in $H$, we add the edges $(Nx + iN/w, Ny)(Nx + (i + 1)N/w, Ny)$ to $L$ for $i = 0, \dots, w - 1$. 
    
        \item Analogously, if a vertical edge $(x, y)(x, y + 1)$ has weight $w$ in $H$, we add the edges $(Nx, Ny + iN/w)(Nx, Ny + (i + 1)N/w)$ for $i = 0, \dots, w - 1$. 
    \end{itemize}
\end{construction}

\begin{df}
    We call the original edges of $L$ \emph{old edges} to distinguish them from the \emph{new edges} we add to form $G$. We call the endpoints of new edges \emph{special vertices}. We call a path \emph{new} if it only passes through new edges. 
\end{df}

\begin{obs} \label{long new edges}
    For any new edge $uv$ in $G$, $d_{\mathbb{Z}^2}(u, v)\geq N/W > 10MW$. 
\end{obs}

For any edge $uv$ in $G$, $d_L(u, v)\leq d_{\mathbb{Z}^2}(u, v)\leq N$. Therefore, $G$ is a bounded extension of $L$. To prove Proposition \ref{sec2 reduction prop}, it suffices to show that $G$ has bounded geodesic slices. 

Throughout this work, we often need to analyse the interaction between geodesics and sublattices of the form $(n\mathbb{Z})^2$. We now introduce the relevant definitions to describe such sublattices. In this section, the sublattice of interest is $(N\mathbb{Z})^2$. To avoid confusion between a point $(x, y)$ in $\mathbb{Z}^2$ defined by its coordinates, and an open interval, we use the notation $]x_1, x_2[$ for the latter. 

\begin{df}
    For a positive integer $n$, let an \emph{open (\emph{resp.}~closed) $n$-interval} be an interval of the form $]nx, nx + n[$ (\emph{resp.}~$[nx, nx + n]$) for some integer $x$. 
\end{df}

\begin{df}
    Let an \emph{open (\emph{resp.}~closed) $n$-block} be a Cartesian product of the form $]nx, nx + n[\times ]ny, ny + n[$ (\emph{resp.}~$[nx, nx + n]\times [ny, ny + n]$) for integers $x, y$. The \emph{boundary} of an $n$-block denotes the difference between its closed and open versions. The boundary consists of four \emph{sides}, two \emph{horizontal sides} to the top and bottom of the $n$-block, and two \emph{vertical sides} to the left and right of the $n$-block. Let an \emph{open (\emph{resp.}~closed) side} exclude (\emph{resp.}~include) the corners in $(n\mathbb{Z})^2$. 
\end{df}

Throughout this paper, an interval always implicitly denotes its intersection with $\mathbb{Z}$. Similarly, a Cartesian product of intervals always implicitly denotes its intersection with the set $\mathbb{Z}^2$ of vertices. 

\begin{df}
    Let an \emph{open (\emph{resp.}~closed) horizontal $n$-strip} be a Cartesian product of the form $\mathbb{Z}\times ]ny, ny + n[$ (\emph{resp.}~$\mathbb{Z}\times [ny, ny + n]$) for some integer $y$. Similarly, let an \emph{open (\emph{resp.}~closed) vertical $n$-strip} be a Cartesian product of the form $]nx, nx + n[\times \mathbb{Z}$ (\emph{resp.}~$[nx, nx + n]\times \mathbb{Z}$) for some integer $x$. 
\end{df}

Figure~\ref{fig:def1} illustrates some examples of $n$-blocks and $n$-strips for $n = 3$. 

\begin{figure}
    \centering
    \begin{tikzpicture}
        \draw[step=1cm,gray!50,very thin] (-1,-1) grid (10, 10);
        
        \foreach \x in {0, 3, 6, 9}{
            \foreach \y in {0, 3, 6, 9}{
                \fill[gray!70] (\x, \y) circle[radius=0.05];
            }
        }
        
        \draw[black,thick] (0,3) -- (0,6); 
        \draw[black,thick] (3,3) -- (3,6);
        \draw[black,thick] (0,6) -- (3,6);
        \draw[black,thick] (0,3) -- (3,3);
        
        \fill[opacity=0.2] (0,3) rectangle (3,6);
        \fill[opacity=0.2] (6,-1) rectangle (9,10);
    \end{tikzpicture}
    \caption{Illustration of an $n$-block (left, shaded grey) with its sides coloured black and a vertical $n$-strip (right, shaded grey) drawn on $\mathbb{Z}^2$ for $n = 3$. Vertices in $(n\mathbb{Z})^2$ are marked in grey.}
    \label{fig:def1}
\end{figure}
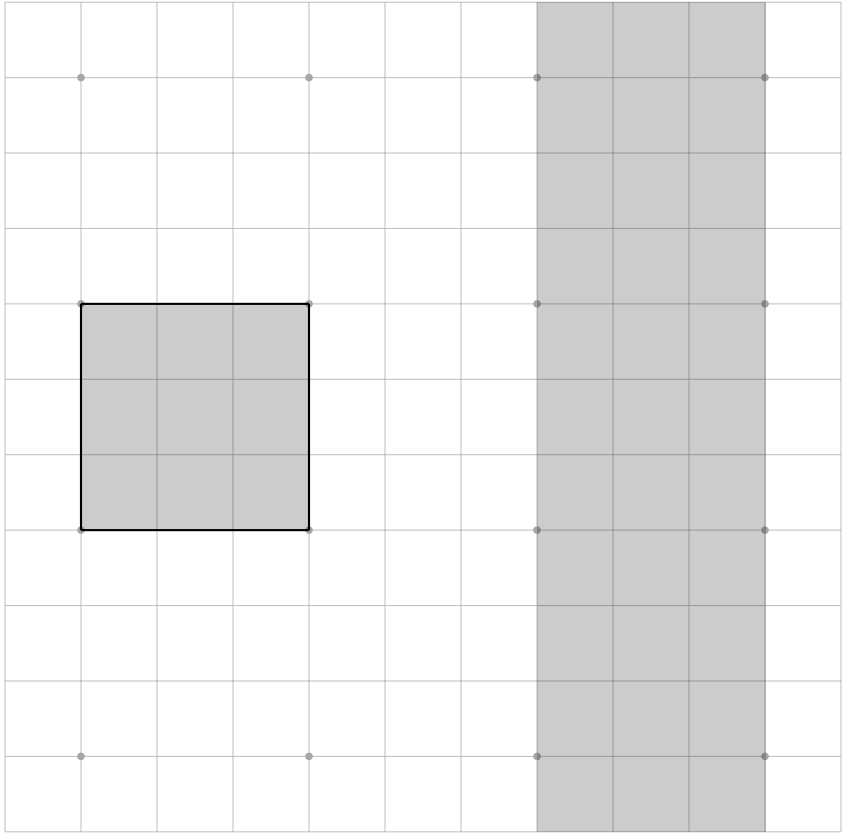

\begin{df}
    For a positive integer $n$ and a vertex $v = (x, y)\in\mathbb{Z}^2$, let $nv = (nx, ny)\in (n\mathbb{Z})^2\subseteq \mathbb{Z}^2$. Conversely, for a vertex $u = (nx, ny)\in (n\mathbb{Z})^2$, let $\frac{1}{n}u = (x, y)\in \mathbb{Z}^2$. 
\end{df}

\begin{df}
    For a positive integer $n$, a vertex $v = (nx, ny)\in (n\mathbb{Z})^2$ is an \emph{$n$-parent} of another vertex $u\in \mathbb{Z}^2$ if $u \in ]nx - n, nx + n[\times ]ny - n, ny + n[$. 
\end{df}

Each vertex in $\mathbb{Z}^2$ has 1, 2 or 4 $n$-parents. We note an immediate characterisation of $n$-parents in terms of closed $n$-blocks. 

\begin{obs} \label{parent characterisation}
    For a positive integer $n$ and vertices $v\in (n\mathbb{Z})^2$, $u\in \mathbb{Z}^2$, $v$ is an $n$-parent of $u$ if and only if every closed $n$-block containing $u$ also contains $v$. 
\end{obs}

\begin{df}
    We call a path in a graph with vertex set $\mathbb{Z}^2$ \emph{straight} if its edges, when traversed in order, are parallel to the same coordinate axis and oriented in the same direction. 
\end{df}

\begin{obs} \label{tight triangle inequality}
    Given a straight path $w_0\dots w_\ell$ in a graph with vertex set $\mathbb{Z}^2$, we have
    \begin{equation*}
        d_{\mathbb{Z}^2}(w_0, w_\ell) = \sum_{i = 1}^{\ell}d_{\mathbb{Z}^2}(w_{i - 1}, w_i). 
    \end{equation*}
    In other words, the triangle inequality is tight here. 
\end{obs}

We now return to the investigation of the metric and geodesic structures of $G$. 

\begin{df}
    For adjacent vertices $v_1, v_2$ in $\mathbb{Z}^2$, let the \emph{short straight path (SSP)} between $Nv_1$ and $Nv_2$ denote the (uniquely existing) straight and new path between $Nv_1$ and $Nv_2$ in $G$. 
\end{df}

The SSPs between vertices in $(N\mathbb{Z})^2$ inherit the weights from $H$. 

\begin{obs} \label{inheritance}
    For an edge $e = v_1v_2$ in $\mathbb{Z}^2$, the SSP between $Nv_1$ and $Nv_2$ is strictly shorter than any other straight path and has length equal to the weight of the edge $e$ in $H$. 
\end{obs}

\begin{lem} \label{rough esti NZ^2}
    For vertices $u, v\in (N\mathbb{Z})^2$, we have
    \begin{equation*}
        d_G(u, v)\leq \frac{W}{N} d_{\mathbb{Z}^2}(u, v)\leq \frac{WM}{N} d_L(u, v) \leq \frac{1}{10M} d_L(u, v). 
    \end{equation*}
\end{lem}

\begin{proof}
    The first inequality follows from Observation~\ref{inheritance} once we connect $u, v$ using $\frac{1}{N}d_{\mathbb{Z}^2}(u, v)$ SSPs. The second inequality follows from Observation~\ref{L vs Z^2}, and the third inequality follows from $N > 10WM^2$. 
\end{proof}

\begin{lem} \label{rough esti line}
    Let $u, v$ be special vertices sharing the same horizontal or vertical coordinate which is a multiple of $N$. Then
    \begin{equation*}
        d_G(u, v)\leq \frac{W}{N} d_{\mathbb{Z}^2}(u, v)\leq \frac{WM}{N} d_L(u, v)\leq \frac{1}{10M} d_L(u, v). 
    \end{equation*}
\end{lem}

\begin{proof}
    A straight and new path exists between $u, v$ in $G$, by virtue of Construction~\ref{reduction pattern}. The first inequality now follows from Observations~\ref{long new edges} and \ref{tight triangle inequality}. Again, the second inequality follows from Observation~\ref{L vs Z^2}, and the third inequality follows from $N > 10WM^2$. 
\end{proof}

\begin{lem} \label{special vert geo}
    A geodesic in $G$ between special vertices only passes through special vertices. In other words, it only uses new edges. 
\end{lem}

\begin{proof}
    For the sake of contradiction, let $\gamma$ be a geodesic between special vertices which passes through a vertex $w'$ which is not special. Let $u, v$ be the nearest special vertices on $\gamma$ to either side of $w'$. We shall consider the sub-geodesic $\gamma'$ of $\gamma$ between $u$ and $v$. Note that the only special vertices that $\gamma'$ passes through are its endpoints $u$ and $v$, and $\gamma'$ only uses old edges. In particular, $d_G(u, v) = d_L(u, v)$. 
    
    Take any $N$-parents $u', v'$ of $u, v$, respectively. Since $u$ and $v$ are special vertices, we have 
    \begin{equation*}
        d_{\mathbb{Z}^2}(u, u'), d_{\mathbb{Z}^2}(v, v') < N. 
    \end{equation*}
    By Lemma~\ref{rough esti NZ^2}, we have
    \begin{equation*}
        d_G(u', v')\leq \frac{1}{10M}d_L(u', v'), 
    \end{equation*}
    and by Lemma~\ref{rough esti line}, we have
    \begin{align*}
        d_G(u, u')\leq \frac{W}{N}d_{\mathbb{Z}^2}(u, u') < W, \\
        d_G(v, v')\leq \frac{W}{N}d_{\mathbb{Z}^2}(v', v) < W. 
    \end{align*}
    By the triangle inequality, we have
    \begin{align*}
        d_G(u, v)&\leq d_G(u, u') + d_G(u', v') + d_G(v', v)\\
        & < 2W + \frac{1}{10M}d_L(u', v')\\
        &\leq 2W + \frac{1}{10M}(d_L(u', u) + d_L(u, v) + d_L(v, v'))\\
        &\leq 2W + \frac{1}{10M}(d_{\mathbb{Z}^2}(u', u) + d_L(u, v) + d_{\mathbb{Z}^2}(v, v'))\\
        & < 2W + \frac{N}{5M} + \frac{1}{10M}d_L(u, v). 
    \end{align*}
    Given that $d_G(u, v) = d_L(u, v)$, it follows that $d_G(u, v) < \frac{10M}{10M - 1}\left(2W + \frac{N}{5M}\right) < \frac{N}{2M}$. Therefore, by Observation~\ref{L vs Z^2},
    \begin{equation*}
        d_{\mathbb{Z}^2}(u, v) \leq Md_L(u, v) < N/2. 
    \end{equation*}
    As $u, v$ are special vertices, they both have a coordinate which is a multiple of $N$. Therefore, $u, v$ either satisfy the assumptions of Lemma~\ref{rough esti line} or $u, v$ share an $N$-parent $w$ with $d_{\mathbb{Z}^2}(u, v) = d_{\mathbb{Z}^2}(u, w) + d_{\mathbb{Z}^2}(w, v)$. In the first case, by Lemma~\ref{rough esti line}, we have
    \begin{equation*}
        d_G(u, v)\leq \frac{1}{10 M}d_L(u, v) \leq d_L(u, v) = d_G(u, v), 
    \end{equation*}
    Hence $d_L(u, v) = 0$ and $u = v$, a contradiction. In the second case, applying Lemma~\ref{rough esti line} to the pairs $u, w$ and $w, v$, and Observation~\ref{L vs Z^2} to the pair $u, v$ yields
    \begin{align*}
        d_G(u, v)\leq d_G(u, w) + d_G(w, v)\leq \frac{W}{N}d_{\mathbb{Z}^2}(u, w) + \frac{W}{N}d_{\mathbb{Z}^2}(w, v) \\
        = \frac{W}{N}d_{\mathbb{Z}^2}(u, v) \leq \frac{WM}{N}d_L(u, v)\leq d_L(u, v) = d_G(u, v), 
    \end{align*}
    again forcing $d_L(u, v) = 0$. Hence $u = v$, a contradiction. 
\end{proof}

Noting that vertices in $(N\mathbb{Z})^2$ are special, Lemma \ref{special vert geo} has the following corollary. 

\begin{cor} \label{Nz2 geos}
    Any geodesic in $G$ between two vertices $u, v\in (N\mathbb{Z})^2$ is a concatenation of SSPs. 
\end{cor}

Combining Corollary~\ref{Nz2 geos} with Observation~\ref{inheritance}, we obtain a one-to-one correspondence between geodesics in $H$ and geodesics in $G$ with endpoints in $(N\mathbb{Z})^2$. 

\begin{lem} \label{inheritance lem}
    For $u, v\in \mathbb{Z}^2$, we have a bijection between geodesics in $H$ between $u, v$ and geodesics in $G$ between $Nu, Nv$, sending a geodesic $u=w_0\dots w_k = v$ in $H$ to the geodesic in $G$ formed by concatenating the SSPs between $Nw_i$ and $Nw_{i + 1}$ in order. In particular, 
    \begin{equation*}
        d_H(u, v) = d_G(Nu, Nv). 
    \end{equation*}
\end{lem}

Having studied the structure of geodesics on special vertices, we now show the relevance of special vertices to general geodesics. 

\begin{lem} \label{geos pass through special verts}
    A geodesic in $G$ of length at least $5N$ must pass through a special vertex. 
\end{lem}

\begin{proof}
    Otherwise, let $\gamma$ be a geodesic in $G$ between vertices $u, v$ of length $\ell\geq 5N$ not passing through any special vertices. Therefore, $\gamma$ uses only old edges and $d_L(u, v) = d_G(u, v) = \ell$. Take any $N$-parent $u'$ of $u$ and any $N$-parent $v'$ of $v$. We have $d_{\mathbb{Z}^2}(u, u'), d_{\mathbb{Z}^2}(v, v') < 2N$. By the triangle inequality, we have
    \begin{equation*}
        d_L(u', v')\leq d_L(u', u) + d_L(u, v) + d_L(v, v') < \ell + 4N. 
    \end{equation*}
    Therefore, by Lemma~\ref{rough esti NZ^2}, we have
    \begin{equation*}
        d_G(u', v')\leq \frac{1}{10M}(\ell + 4N). 
    \end{equation*}
    Hence, by the triangle inequality, 
    \begin{equation*}
        \ell = d_G(u, v)\leq d_G(u, u') + d_G(u', v') + d_G(v', v) < 4N + \frac{1}{10M}(\ell + 4N), 
    \end{equation*}
    which contradicts the assumption that $\ell \geq 5N$. 
\end{proof}

We now prove that any sufficiently long geodesic in $G$ is composed of a geodesic between vertices in $(N\mathbb{Z})^2$ with two short paths, one on each end. 

\begin{lem} \label{all roads pass through Nz2}
    Let $u, v\in\mathbb{Z}^2$ be vertices with $d_G(u, v)\geq 11N$. Any geodesic $\gamma$ between $u, v$ must pass through $u', v'\in (N\mathbb{Z})^2$ with $d_G(u, u'), d_G(v, v') < 6N$. 
\end{lem}

\begin{proof}
    By Lemma~\ref{geos pass through special verts}, any segment of $\gamma$ of length at least $5N$ contains a special vertex. Therefore, $\gamma$ passes through special vertices $u'', v''$ with $d_G(u, u''), d_G(v, v'') \leq 5N$. 
    
    By Lemma~\ref{special vert geo}, the segment of $\gamma$ between $u'', v''$ only uses new edges. As this segment has length at least $N\geq W$, it must pass through $u', v'\in (N\mathbb{Z})^2$ with $d_G(u'', u'), d_G(v'', v') < W$. Therefore, by the triangle inequality, $d_G(u, u'), d_G(v, v') < 6N$. 
\end{proof}

We are now ready to give a proof of Proposition \ref{sec2 reduction prop}. 

\begin{proof}[Proof of Proposition~\ref{sec2 reduction prop}]
    As remarked after Construction~\ref{reduction pattern}, it suffices to show that $G$ has bounded geodesic slices. Let geodesics slices in $H$ have sizes bounded above by $B$. We shall show that geodesic slices of $G$ have sizes bounded above by $10^6N^{11}B$. Note that $N, B$ are both constants dependent only on the graph $L$ and the weighted graph $H$. 

    Fix any pair $u, v\in \mathbb{Z}^2$ of vertices. If $d_G(u, v) < 11N$, then any vertex $w$ on a geodesic from $u$ to $v$ must have $d_G(u, w)\leq d_G(u, v) < 11N$. Therefore, $d_{\mathbb{Z}^2}(u, w)\leq Nd_G(u, w) < 11 N^2$. Therefore, there are at most $5(11N^2)^2\leq 1000N^4$ such vertices $w$. In particular, any geodesic slice $S(u, v, k)$ has size at most $1000N^4$. 

    Therefore, we now assume that $d_G(u, v)\geq 11N$. Fix any non-negative integer $k\leq d_G(u, v)$. Let $U = (N\mathbb{Z})^2\cap B_G(u, 6N)\subseteq B_{\mathbb{Z}^2}(u, 6N^2)$ and $V = (N\mathbb{Z})^2\cap B_G(v, 6N)\subseteq B_{\mathbb{Z}^2}(v, 6N^2)$. We have $|U|, |V|\leq 5(6N^2)^2\leq 200N^4$. For $u'\in U, v'\in V$ and integer $k'\in K_{u'} := [k - d_G(u, u'), k - d_G(u, u') + W] $, let $T_{u', v', k'} = \left\{Nw\middle|w\in S_H\left(\frac{1}{N}u', \frac{1}{N}v', k'\right)\right\}$. Let 
    \begin{equation*}
        T = \bigcup_{\substack{u'\in U, v'\in V\\k'\in K_{u'}}}T_{u', v', k'}. 
    \end{equation*}
    Note that we have $|T|\leq B|U||V|(W + 1)\leq 10^5N^9B$. Let $R = \left(\bigcup_{w\in T} B_{\mathbb{Z}^2}(w, N) \right)\cup B_{\mathbb{Z}^2}(u, 6N^2)\cup B_{\mathbb{Z}^2}(v, 6N^2)$. We have
    \begin{equation*}
        |R|\leq 5N^2|T| + 1000N^4\leq 10^6N^{11}B. 
    \end{equation*}
    We claim that $S(u, v, k)\subseteq R$. If not, let $w\in S(u, v, k)\backslash R$ lie on a geodesic $\gamma$ from $u$ to $v$. By Lemma \ref{all roads pass through Nz2}, as $d_G(u, v)\geq 11N$, $\gamma$ passes through $u', v'\in (N\mathbb{Z})^2$ with $d_G(u, u'), d_G(v, v') < 6N$. Since $w\notin R$, $d_{\mathbb{Z}^2}(u, w), d_{\mathbb{Z}^2}(v, w) > 6N^2$, and hence $d_G(u, w), d_G(v, w) > 6N$. Therefore, $w$ lies on the segment $\gamma'$ of $\gamma$ between $u', v'$. By Lemma \ref{inheritance lem}, there exists a geodesic $\frac{1}{N}u' = \Tilde{w}_0\dots \Tilde{w}_l = \frac{1}{N}v'$ in $H$, for which $\gamma'$ is the concatenation of the SSPs between $N\Tilde{w}_i$ and $N\Tilde{w}_{i + 1}$, in order. Let $w$ lie on the SSP between $N\Tilde{w}_i$ and $N\Tilde{w}_{i + 1}$. As 
    \begin{equation*}
        0\leq d_G(w, N\Tilde{w}_{i + 1})\leq d_G(N\Tilde{w}_i, N\Tilde{w}_{i + 1}) = d_H(\Tilde{w}_i, \Tilde{w}_{i + 1})\leq W, 
    \end{equation*}
    we have
    \begin{equation*}
        d_H(\Tilde{w}_0, \Tilde{w}_{i + 1}) = d_G(u', N\Tilde{w}_{i + 1}) = d_G(u', w) + d_G(w, N\Tilde{w}_{i + 1}) = k - d_G(u, u') + d_G(w, N\Tilde{w}_{i + 1})\in K_{u'}. 
    \end{equation*}
    Hence, $N\Tilde{w}_{i + 1}\in T_{u', v', d_H(\Tilde{w}_0, \Tilde{w}_{i + 1})}\subseteq T$. As $d_{\mathbb{Z}^2}(w, N\Tilde{w}_{i + 1})\leq N$, we have $w\in R$ as needed, a contradiction. Therefore, $S(u, v, k)\subseteq R$, and $|S(u, v, k)|\leq |R|\leq 10^6N^{11}B$. 
\end{proof}

Proposition \ref{sec2 reduction prop} reduces Theorem \ref{z2 gen ext thm} to the following proposition. 

\begin{prop} \label{main reduced prop}
    There exists a weighted graph $H$ on $\mathbb{Z}^2$ using finitely many positive integer weights with bounded geodesic slices. 
\end{prop}

Sections \ref{construction sec}, \ref{fractal} and \ref{proof} are dedicated to the proof of Proposition \ref{main reduced prop}. 

\section{The construction} \label{construction sec}

Our construction of the weighted graph $H$ requires us to fix an odd integer parameter $p\geq 3$. Any constant choice of $p$ shall suffice in yielding a weighted graph $H$ with two positive integer weights and bounded geodesic slices. Heuristically, smaller values of $p$ result in tighter upper bounds on the sizes of geodesic slices. 

\begin{df}
    An edge in $\mathbb{Z}^2$ is called a \emph{slow edge} if it is of the form 
    \begin{equation*}
        \left(p^n x + \frac{p^n - 1}{2}, p^n y\right)\left(p^n x + \frac{p^n + 1}{2}, p^n y\right) \quad \text{or}\quad\left(p^n x, p^n y + \frac{p^n - 1}{2}\right)\left(p^n x, p^n y + \frac{p^n + 1}{2}\right)
    \end{equation*}
    for non-negative integers $n$ and integers $x, y$ not divisible by $p$. All other edges of $\mathbb{Z}^2$ are called \emph{fast edges}. 
\end{df}

\begin{df}
    For positive integers $a < b$, let $H_{a, b}$ be the weighted graph on $\mathbb{Z}^2$, where we give the weight $a$ to fast edges and $b$ to slow edges. 
\end{df}

We note that the weights of $H_{a, b}$ are symmetric under reflection across $x = y$. 

\begin{obs}
    $(x, y)(x + 1, y)$ is a slow edge if and only if $(y, x)(y, x + 1)$ is a slow edge. 
\end{obs}

We further note that, for any pair of vertices $u, v$, we have
\begin{equation*}
    ad_{\mathbb{Z}^2}(u, v)\leq d_{H_{a, b}}(u, v)\leq bd_{\mathbb{Z}^2}(u, v), 
\end{equation*}
as the minimum and maximum weights of $H_{a, b}$ are $a$ and $b$, respectively. 

\section{A fractal property} \label{fractal}

We start by examining the role of the sublattice $(p\mathbb{Z})^2$ in the geodesic structure of $H_{a, b}$. Informally, we show that $(p\mathbb{Z})^2$ in $H_{a, b}$ emulates $H_{pa, pa - a + b}$. The analysis here closely mirrors the arguments in Section \ref{intro and weight}, where we related the geodesic structure of the unweighted graph $G$ on $(N\mathbb{Z})^2$ to the weighted graph $H$. 

We start by observing two properties of the sublattice $(p\mathbb{Z})^2$ in the weighted graphs $H_{a, b}$. 

\begin{obs} \label{pz2 fast}
    All edges incident to a vertex in $(p\mathbb{Z})^2$ are fast. 
\end{obs}

\begin{obs} \label{fractal inheritance}
    For an edge $e = v_1v_2$ in $\mathbb{Z}^2$, the straight path between $pv_1$ and $pv_2$ has cost $pa$ if $e$ is fast and $pa - a + b$ if $e$ is slow. 
\end{obs}

\begin{df}
    In analogy with the special vertices of Section \ref{intro and weight}, we call a vertex of $\mathbb{Z}^2$ \emph{horizontally (\emph{resp.}~vertically) distinguished} if its horizontal (\emph{resp.}~vertical) coordinate is a multiple of $p$. We call a vertex \emph{distinguished} if it is either horizontally or vertically distinguished. 
\end{df}

Note that a vertex is distinguished if and only if it does not lie in an open $p$-block. 

We now seek to prove an analogue (Lemma~\ref{unaligned distinguished endpoints geo}) of Lemma \ref{special vert geo} for geodesics between distinguished vertices. However, extra assumptions on the endpoints are needed here. 

\begin{df}
    A pair of distinguished vertices are \emph{horizontally (\emph{resp.}~vertically) aligned} if their vertical (\emph{resp.}~horizontal) coordinates lie in the same open $p$-interval. A pair of distinguished vertices are \emph{aligned} if they are horizontally or vertically aligned. A pair of distinguished vertices are said to be \emph{unaligned} if they are not aligned. 
\end{df}

Figure~\ref{fig:def2} illustrates a pair of horizontally aligned distinguished vertices for $p = 3$. 

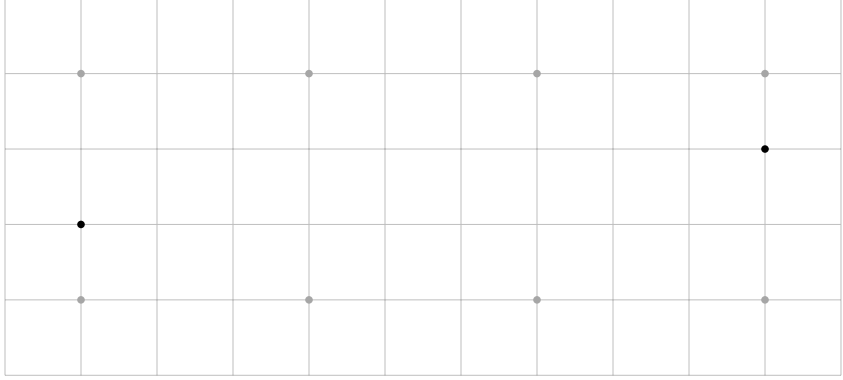
\begin{figure}
    \centering
    \begin{tikzpicture}
        \draw[step=1cm,gray!50,very thin] (-1,-1) grid (10, 4);
        
        \foreach \x in {0, 3, 6, 9}{
            \foreach \y in {0, 3}{
                \fill[gray!70] (\x, \y) circle[radius=0.05];
            }
        }

        \fill (0, 1) circle[radius=0.05];   
        \fill (9, 2) circle[radius=0.05];   
    \end{tikzpicture}
    \caption{Illustration of a pair of horizontally aligned vertices drawn in black for $p = 3$. Vertices in $(p\mathbb{Z})^2$ are coloured grey.}
    \label{fig:def2}
\end{figure}

Note that vertices in a horizontally aligned pair have vertical coordinates which are not multiples of $p$. Therefore, only a pair of horizontally (and not vertically) distinguished vertices can be horizontally aligned. Similarly, only a pair of vertically (and not horizontally) distinguished vertices can be vertically aligned. 

\begin{df}
    For a path $\gamma$ and a subset $A$ of vertices, let \emph{an excursion of $\gamma$ in $A$} be a maximal segment of $\gamma$ contained in $A$. We call an excursion $\gamma'$ of $\gamma$ in $A$ \emph{internal} if it does not contain either endpoint of $\gamma$. Let the \emph{entry and exit vertices} of an internal excursion $\gamma'$ denote the vertices on $\gamma$ immediately preceding and following $\gamma'$. 
\end{df}

\begin{lem} \label{opposite sides}
    Let $\gamma$ be a geodesic in $H_{a, b}$. Let $\gamma'$ be an internal excursion of $\gamma$ in an open $p$-block, with entry and exit vertices $u, v$. Then $u, v$ must lie on opposite open sides of the $p$-block. In particular, $d_{H_{a, b}}(u, v)\geq pb$. 
\end{lem}

\begin{proof}
    It suffices to show that $u, v$ cannot lie on the same closed side, or two adjacent closed sides of the $p$-block $[px, px + p]\times [py, py + p]$. 
    
    If $u, v$ lie on the same closed side, we may assume, without loss of generality, that $u = (px + i, py), v = (px + j, py)$, where $0 < i < j < p$. As the segment of $\gamma$ between $u, v$ enters the open $p$-block, it uses only slow edges and has length at least $j - i + 2$, and therefore costs at least $(j - i + 2)b$. On the other hand, the straight path between $u, v$ has length $j - i$, and hence costs at most $(j - i)b < (j - i + 2)b$, a contradiction. 

    If $u, v$ lie on the adjacent closed sides, we may assume, without loss of generality, that $u = (px + i, py), v = (px, py + j)$, where $0 < i, j < p$. As the segment of $\gamma$ between $u, v$ enters the open $p$-block, it uses only slow edges and has length at least $i + j$, and therefore costs at least $(i + j)b$. On the other hand, the concatenation of the straight paths between $u, (px, py)$ and between $(px, py), v$ has length $i + j$, and uses at least two fast edges by Observation \ref{pz2 fast}. Therefore, this concatenation of straight paths costs at most $(i + j - 2)b + 2a < (i + j)b$, a contradiction. 
\end{proof}

\begin{cor} \label{one excursion}
    A geodesic $\gamma$ in $H_{a, b}$ cannot have more than one internal excursion in an open $p$-block. 
\end{cor}

\begin{proof}
    For the sake of contradiction, let $\gamma$ have at least two internal excursions $\gamma_1, \gamma_2$ to an open $p$-block $]px, px + p[\times ]py, py + p[$. Let $\gamma_1$ precede $\gamma_2$ and let $u_i, v_i$ be the entry and exit vertices of $\gamma_i$. Each vertex on the boundary of a $p$-block has at most one neighbour in the open $p$-block. Therefore, we must have $v_2\neq u_1$, otherwise, $v_2 = u_1$ is preceded and followed in $\gamma$ by the same vertex. Hence, by Lemma~\ref{opposite sides}, we have
    \begin{equation*}
        d_{H_{a, b}}(u_1, v_2) = d_{H_{a, b}}(u_1, v_1) + d_{H_{a, b}}(v_1, u_2) + d_{H_{a, b}}(u_2, v_2)\geq 2pb + d_{H_{a, b}}(v_1, u_2) > 2pb. 
    \end{equation*}
    On the other hand, 
    \begin{equation*}
        d_{H_{a, b}}(u_1, v_2) \leq d_{\mathbb{Z}^2}(u_1, v_2)b\leq 2pb, 
    \end{equation*}
    a contradiction. 
\end{proof}

\begin{lem} \label{unaligned distinguished endpoints geo}
    In $H_{a, b}$, geodesics between unaligned distinguished vertices only pass through distinguished vertices. 
\end{lem}

\begin{proof}
    Recall that a vertex is distinguished if and only if it does not lie in an open $p$-block. 
    
    For the sake of contradiction, let $\gamma$ be a geodesic between unaligned distinguished endpoints $u, v$ passing through an open $p$-block. We further take $\gamma$ to have minimal length amongst such geodesics. 
    
    If $\gamma$ passes through any intermediate distinguished vertex $w$ which is unaligned with both $u$ and $v$, we may consider the segments $\gamma_1$ and $\gamma_2$ of $\gamma$ between $u, w$ and $w, v$, respectively. As they are strictly shorter than $\gamma$ and have unaligned distinguished endpoints, they must only pass through distinguished vertices, a contradiction. 
    
    Therefore, $\gamma$ cannot pass through any intermediate distinguished vertex $w$ which is unaligned with both $u$ and $v$. In particular, $\gamma$ cannot pass through any intermediate vertices in $(p\mathbb{Z})^2$. 

    By Lemma \ref{opposite sides} and Corollary \ref{one excursion}, without visiting intermediate vertices in $(p\mathbb{Z})^2$, $\gamma$ enters and exits open $p$-blocks either exclusively on horizontal sides or exclusively on vertical sides. Without loss of generality, we assume the former. Therefore, we have a vertical sequence of $p$-blocks $]px, px + p[\times ]py, py + p[, ]px, px + p[\times ]py + p, py + 2p[, \dots, ]px, px + p[\times ]py + (c - 1)p, py + cp[$ which $\gamma$ traverses in order. As $\gamma$ is assumed to pass through at least one open $p$-block, we have $c\geq 1$. The endpoints $u, v$ must lie in $[px, px + p]\times\{py\}$ and $[px, px + p]\times\{py + cp\}$, respectively. For $u$ and $v$ to not be horizontally aligned, at least one of them must lie in $(p\mathbb{Z})^2$. 

    Without loss of generality, let $u = (px, py)\in (p\mathbb{Z})^2$, and let $v = (px + i, py + cp)$. $\gamma$ has length at least $cp + i$ and all vertical edges in $\gamma$ are slow. Therefore, $\gamma$ costs at least $cpb + ia$. On the other hand, the concatenation of the straight paths between $u, (px, py + cp)$ and between $(px, py + cp), v$ has length $cp + i$ and uses at most $c + 1$ slow edges. Therefore, this concatenation of straight paths costs at most $(c + 1)b + (cp + i - c - 1)a < cpb + ia$, a contradiction. 
\end{proof}

\begin{df}
    Let an $n$-straight path ($n$-SP) denote the straight path between $nu$ and $nv$, where $u, v$ are adjacent vertices in $\mathbb{Z}^2$. 
\end{df}

\begin{df}
    For a positive integer $n$, and a path $\gamma = w_0\dots w_\ell$ in $\mathbb{Z}^2$, let the \emph{$n$-dilation} of $\gamma$ denote the path formed by concatenating the $n$-SPs between $nw_i, nw_{i + 1}$, in order. 
\end{df}

A pair of distinct vertices in $(p\mathbb{Z})^2$ form an unaligned pair of distinguished vertices. Therefore, Lemma~\ref{unaligned distinguished endpoints geo} has the following corollary, in parallel to Corollary~\ref{Nz2 geos}. 

\begin{cor} \label{pz2 geos}
    A geodesic in $H_{a, b}$ between two vertices $u, v\in (p\mathbb{Z})^2$ only passes through distinguished vertices, and therefore is a concatenation of $p$-SPs. 
\end{cor}

Combining Observation~\ref{fractal inheritance} with Corollary~\ref{pz2 geos}, we obtain (in parallel to Lemma~\ref{inheritance lem}) a one-to-one correspondence between geodesics in $H_{pa, pa - a + b}$ and geodesics in $H_{a, b}$ with endpoints in $(p\mathbb{Z})^2$. 

\begin{lem} \label{fractal inheritance lem}
    For $u, v\in \mathbb{Z}^2$, we have a bijection from geodesics in $H_{pa, pa - a + b}$ between $u, v$ to geodesics in $H_{a, b}$ between $pu, pv$, sending a geodesic $\gamma$ in $H_{pa, pa - a + b}$ to its $p$-dilation $\gamma'$. In particular, as the cost of $\gamma$ in $H_{pa, pa - a + b}$ is equal to the cost of $\gamma'$ in $H_{a, b}$, we have
    \begin{equation*}
        d_{H_{pa, pa - a + b}}(u, v) = d_{H_{a, b}}(pu, pv). 
    \end{equation*}
\end{lem}

Iterating Lemma~\ref{fractal inheritance lem}, we may generalise it to any power $p^n$ of $p$. 

\begin{lem} \label{fractal inheritance lem 2}
    For $u, v\in \mathbb{Z}^2$, we have a bijection from geodesics in $H_{p^na, p^na - a + b}$ between $u, v$ to geodesics in $H_{a, b}$ between $p^nu, p^nv$, sending a geodesic $\gamma$ in $H_{p^na, p^na - a + b}$ to its $p^n$-dilation $\gamma'$. In particular, as the cost of $\gamma$ in $H_{p^na, p^na - a + b}$ is equal to the cost of $\gamma'$ in $H_{a, b}$, we have
    \begin{equation*}
        d_{H_{p^na, p^na - a + b}}(u, v) = d_{H_{a, b}}(p^nu, p^nv). 
    \end{equation*}
\end{lem}

We finish this section by studying the local behaviour of geodesics within $p^n$-blocks. 

\begin{lem} \label{loc behaviour}
    Let $\gamma$ be a geodesic in $H_{a, b}$ between vertices $u, v\in\mathbb{Z}^2$. If $v\in (p\mathbb{Z})^2$ and $v$ is the only vertex in $(p\mathbb{Z})^2$ that $\gamma$ passes through, then any closed $p$-block $B$ containing $u$ must contain the entire geodesic $\gamma$. 

    In particular, $v$ is a $p$-parent of $u$ and $\gamma$ is contained within a single closed $p$-block. 
\end{lem}

\begin{proof}
    Let $w$ be the first distinguished vertex we encounter as we traverse along $\gamma$ from $u$ to $v$. Let $B$ be any closed $p$-block containing $u$. 
    
    As $\gamma$ cannot leave a closed $p$-block without encountering a distinguished vertex and $w$ is the first distinguished vertex encountered by $\gamma$, $B$ must contain $w$. Since $v\in (p\mathbb{Z})^2$, $v, w$ are unaligned distinguished vertices. By Lemma~\ref{unaligned distinguished endpoints geo}, the segment $\gamma'$ of $\gamma$ from $w$ to $v$ only passes through distinguished vertices. Hence $\gamma'$ can only leave a closed $p$-block at one of its corners. Therefore, $\gamma$ never leaves $B$. 
    
    The assertion that $v$ is a $p$-parent of $u$ follows from Observation~\ref{parent characterisation}. 
\end{proof}

Using Lemma~\ref{fractal inheritance lem 2}, we generalise Lemma~\ref{loc behaviour} to all powers $p^n$ of $p$. 

\begin{cor} \label{loc behaviour cor}
    Let $\gamma$ be a geodesic in $H_{a, b}$ between vertices $u, v\in\mathbb{Z}^2$. If $v\in (p^n\mathbb{Z})^2$ and $v$ is the only vertex in $(p^n\mathbb{Z})^2$ that $\gamma$ passes through, then any closed $p^n$-block $B$ containing $u$ must contain the entire geodesic $\gamma$. 

    In particular, $v$ is a $p^n$-parent of $u$ and $\gamma$ is contained within a single closed $p^n$-block. 
\end{cor}

\begin{proof}
    We induct on $n$. The base case $n = 0$ is trivial. Assuming the statement holds for $n$, we now show it for $n + 1$. 

    Let $w$ be the first vertex in $(p^n\mathbb{Z})^2$ we encounter as we traverse $\gamma$ from $u$ to $v$. Applying the inductive hypothesis to the segment $\gamma_1$ of $\gamma$ from $u$ to $w$, any closed $p^n$-block containing $u$ must contain $\gamma_1$. In particular, any closed $p^{n + 1}$-block containing $u$ must contain $\gamma_1$. 
    
    By Lemma~\ref{fractal inheritance lem 2}, the segment $\gamma_2$ of $\gamma$ from $w$ to $v$ is a $p^n$-dilation of a geodesic $\gamma'$ in $H_{p^na, p^na - a + b}$ from $\frac{1}{p^n}w$ to $\frac{1}{p^n}v$. Since $v$ is the only vertex in $(p^{n + 1}\mathbb{Z})^2$ on $\gamma_2$, $\frac{1}{p^n}v$ is the only vertex in $(p\mathbb{Z})^2$ on $\gamma'$. Applying Lemma~\ref{loc behaviour} to $\gamma'$, any closed $p$-block containing $\frac{1}{p^n}w$ must contain $\gamma'$. Therefore, any closed $p^{n + 1}$-block containing $w$ must contain $\gamma_2$. 

    Combining these observations, any closed $p^{n + 1}$-block $B$ containing $u$ must contain $\gamma_1$, and in particular $w$. Therefore, $B$ must also contain $\gamma_2$, and, hence, the entirety of $\gamma$. The assertion that $v$ is a $p^{n + 1}$-parent of $u$ again follows from Observation~\ref{parent characterisation}. 
\end{proof}

\section{Bounded geodesic slices} \label{proof}

\begin{lem} \label{classification level 1}
    Let $\gamma$ be a geodesic in $H_{a, b}$ between vertices $u, v\in\mathbb{Z}^2$. Then the following statements are equivalent. 
    \begin{enumerate} [label = (\alph*)]
        \item $\gamma$ is not contained in any open (horizontal or vertical) $p$-strip.  
        \item $\gamma$ passes through a vertex in $(p\mathbb{Z})^2$. 
    \end{enumerate}
\end{lem}

\begin{proof}
    (b) implies (a) trivially, as any vertex in $(p\mathbb{Z})^2$ is not contained in any open $p$-strip. 
    
    We now show that (a) implies (b). If $\gamma$ is not contained in any open $p$-strip, then $\gamma$ must contain both a horizontally distinguished vertex $w_H$ and a vertically distinguished vertex $w_V$. Note that $w_H$ and $w_V$ are unaligned. Therefore, by Lemma~\ref{unaligned distinguished endpoints geo}, the segment $\gamma'$ of $\gamma$ between $w_H$ and $w_V$ only passes through distinguished vertices. However, one cannot travel from a horizontally distinguished vertex to a vertically distinguished vertex without passing through a vertex in $(p\mathbb{Z})^2$ or a vertex which is not distinguished. As a result, $\gamma'$ must pass through a vertex in $(p\mathbb{Z})^2$. In particular, $\gamma$ must pass through a vertex in $(p\mathbb{Z})^2$. 
\end{proof}

We may, in essence, iterate Lemma~\ref{classification level 1} 
to generalise it to any power $p^n$ of $p$. 

\begin{lem} \label{classification level n}
    Let $\gamma$ be a geodesic in $H_{a, b}$ between vertices $u, v\in\mathbb{Z}^2$. For any non-negative integer $n$, the following statements are equivalent. 
    \begin{enumerate} [label = (\alph*)]
        \item $\gamma$ is not contained in any open $p^n$-strip. \item $\gamma$ passes through a vertex in $(p^n\mathbb{Z})^2$. 
    \end{enumerate}
\end{lem}

\begin{proof}
    Again, (b) implies (a) trivially. We prove that (a) implies (b) by induction on $n$. The base case $n = 0$ is trivial. Assuming the statement holds for some $n\geq 0$, we now show it for $n + 1$. 

    As $\gamma$ is not contained in any open $p^{n + 1}$-strip, it is, in particular, not contained in any open $p$-strip. By Lemma~\ref{classification level 1}, $\gamma$ passes through a vertex in $(p\mathbb{Z})^2$. As we traverse $\gamma$ from $u$ to $v$, let $u', v'$ be the first and last vertices in $(p\mathbb{Z})^2$ we encounter, respectively. By Lemma~\ref{loc behaviour}, $u', v'$ are $p$-parents of $u, v$, respectively. Hence, the segment $\gamma'$ of $\gamma$ between $u', v'$ is not contained in any open $p^{n + 1}$-strip. 
    
    By Lemma~\ref{fractal inheritance lem}, $\gamma'$ is the $p$-dilation of a geodesic $\gamma''$ in $H_{pa, pa - a + b}$ from $\frac{1}{p}u'$ to $\frac{1}{p}v'$. As $\gamma'$ is not contained in any open $p^{n + 1}$-strip, $\gamma''$ is not contained in any open $p^n$-strip. Applying the inductive hypothesis to $\gamma''$ shows that $\gamma''$ must pass through a vertex in $(p^n\mathbb{Z})^2$. As a result, $\gamma'$ must pass through a vertex in $(p^{n + 1}\mathbb{Z})^2$. In particular, $\gamma$ must pass through a vertex in $(p^{n + 1}\mathbb{Z})^2$. 
\end{proof}

\begin{df}
    Let $\gamma$ be a path in $\mathbb{Z}^2$ between vertices $u, v$. If $\gamma$ does not pass through the origin $(0, 0)$, let $n(\gamma)$ be the largest non-negative integer for which $\gamma$ passes through a vertex in $(p^{n(\gamma)}\mathbb{Z})^2$; and let $\Tilde{u}(\gamma), \Tilde{v}(\gamma)$ be the first and last vertices in $(p^{n(\gamma)}\mathbb{Z})^2$ encountered as we traverse $\gamma$ from $u$ to $v$. If $\gamma$ passes through the origin $(0, 0)$, let $n(\gamma) = \infty$ and let $\Tilde{u}(\gamma) = \Tilde{v}(\gamma) = (0, 0)$. 
\end{df}

\begin{obs} \label{cut points observation}
    If $\gamma$ is a geodesic in $H_{a, b}$ from $u$ to $v$ and $\gamma$ does not pass through the origin, then, by Corollary~\ref{loc behaviour cor}, $\Tilde{u}(\gamma), \Tilde{v}(\gamma)$ are $p^{n(\gamma)}$-parents of $u, v$, respectively. 
\end{obs}

\begin{df}
    Let $\Tilde{u}(\gamma)$ and $\Tilde{v}(\gamma)$ divide $\gamma$ into three (potentially empty) segments: $s_1(\gamma)$ between $u, \Tilde{u}(\gamma)$, $s_2(\gamma)$ between $\Tilde{u}(\gamma), \Tilde{v}(\gamma)$ and $s_3(\gamma)$ between $\Tilde{v}(\gamma), v$. 
\end{df}

\begin{df}
    For vertices $u, v$ in $H_{a, b}$, an integer $k$, and $i = 1, 2, 3$, let the \emph{partial geodesic slice} $S^{(i)}(u, v, k)\subseteq S(u, v, k)$ denote the set of vertices $w$ lying on the segment $s_i(\gamma)$ of a geodesic $\gamma$ from $u$ to $v$, and satisfying $d_{H_{a, b}}(u, w) = k$. 
\end{df}

Since each geodesic $\gamma$ from $u$ to $v$ is covered by the segments $s_1(\gamma), s_2(\gamma), s_3(\gamma)$, the geodesic slice $S(u, v, k)$ is the union $\bigcup_{i = 1, 2, 3} S^{(i)}(u, v, k)$ of the three partial geodesic slices. Therefore, it suffices to upper bound the sizes of the partial geodesic slices. We shall upper bound the size of $S^{(1)}(u, v, k)$ (and, by symmetry, of $S^{(3)}(u, v, k)$) in Lemma~\ref{s1} and the size of $S^{(2)}(u, v, k)$ in Lemma~\ref{s2}. 

\begin{lem} \label{vp increases quickly}
    Let $\gamma$ be a geodesic in $H_{a, b}$ between vertices $u, v$ and let $n$ be a non-negative integer. If $s_1(\gamma)$ has cost at least $2bp^n$, then $n(\gamma) > n$, and hence the endpoint $\Tilde{u}(\gamma)$ of $s_1(\gamma)$ is in $(p^{n + 1}\mathbb{Z})^2$. 
\end{lem}

\begin{proof}
    If $\gamma$ contains $(0, 0)$, then $n(\gamma) = \infty > n$, and $\Tilde{u}(\gamma) = (0, 0)\in (p^{n + 1}\mathbb{Z})^2$. 
    
    Otherwise, by Observation~\ref{cut points observation}, the endpoint $\Tilde{u}(\gamma)$ of $s_1(\gamma)$ is a $p^{n(\gamma)}$-parent of $u$. Therefore, $d_{H_{a, b}}(u, \Tilde{u}(\gamma))\leq bd_{\mathbb{Z}^2}(u, \Tilde{u}(\gamma)) < 2bp^{n(\gamma)}$. As the cost $d_{H_{a, b}}(u, \Tilde{u}(\gamma))$ of $s_1(\gamma)$ is at least $2bp^n$, we must have $n(\gamma) > n$. 
\end{proof}

\begin{lem} \label{s1}
    For all vertices $u, v$ in $H_{a, b}$ and all integers $k$, we have
    \begin{equation*}
        \left|S^{(1)}(u, v, k)\right|\leq 80\frac{b^2p^2}{a^2}.
    \end{equation*}
\end{lem}

\begin{proof}
    For $k\leq 2b$, we have $S^{(1)}(u, v, k)\subseteq S(u, v, k)\subseteq B_{H_{a, b}}(u, k)\subseteq B_{\mathbb{Z}^2}(u, 2b/a)$. Therefore, $|S^{(1)}(u, v, k)|\leq |B_{\mathbb{Z}^2}(u, 2b/a)|\leq 5(2b/a)^2$. Hence, we now assume that $k > 2b$. Let $n$ be the largest non-negative integer such that $2bp^n\leq k$. We first show that every $w\in S^{(1)}(u, v, k)$ lies on a $p^n$-SP $\lambda$ on which $w$ is the only vertex in $S^{(1)}(u, v, k)$. 

    Let $w$ lie on $s_1(\gamma)$ for a geodesic $\gamma$ from $u$ to $v$. Denote the endpoint $\Tilde{u}(\gamma)$ of $s_1(\gamma)$ by $u'$. The cost of $s_1(\gamma)$ is at least $d_{H_{a, b}}(u, w) = k \geq 2bp^n$. Therefore, by Lemma~\ref{vp increases quickly}, $u'\in (p^{n + 1}\mathbb{Z})^2\subseteq (p^n\mathbb{Z})^2$. Let $u''$ be the first vertex in $(p^n\mathbb{Z})^2$ encountered as we traverse $s_1(\gamma)$ from $u$ to $u'$. By Corollary~\ref{loc behaviour cor}, $u''$ is a $p^n$-parent of $u$. Therefore, $d_{H_{a, b}}(u, u'')\leq bd_{\mathbb{Z}^2}(u, u'') < 2bp^n \leq k = d_{H_{a, b}}(u, w)$. As a result, $w$ lies on the segment $\gamma'$ of $\gamma$ between $u''$ and $u'$. By Lemma~\ref{fractal inheritance lem 2}, $\gamma'$ is composed of $p^n$-SPs. Therefore, $w$ lies on a $p^n$-SP $\lambda$ which is a segment of the geodesic $\gamma$. Therefore, all vertices in $\lambda$ have different distance to $u$ in $H_{a, b}$. Hence, $w$ is the only vertex on $\lambda$ in $S^{(1)}(u, v, k)$. 

    As $\lambda$ is a segment of $\gamma$, one of the endpoints $w'$ of $\lambda$ lies in $B_{H_{a, b}}(u, k)\subseteq B_{H_{a, b}}(u, 2bp^{n + 1})\subseteq B_{\mathbb{Z}^2}(u, 2bp^{n + 1}/a)$. As $w$ is determined by $\lambda$, and there are at most four choices of $\lambda$ given $w'$, we have
    \begin{equation*}
        |S^{(1)}(u, v, k)|\leq 4 |B_{\mathbb{Z}^2}(u, 2bp^{n + 1}/a)\cap (p^n\mathbb{Z})^2|\leq 20 (2b/a)^2 p^2. \qedhere
    \end{equation*}
\end{proof}

To upper bound the size of $S^{(2)}(u, v, k)$, we first further express $S^{(2)}(u, v, k)$ as a union. 

\begin{df}
    Given vertices $u, v$ in $H_{a, b}$ and $n\in\mathbb{Z}^{\geq 0}\cup\{\infty\}$, let $\Gamma(u, v, n)$ denote the set of geodesics $\gamma$ from $u$ to $v$ with $n(\gamma) = n$. For vertices $u, v$ in $H_{a, b}$, an integer $k$, and $n\in\mathbb{Z}^{\geq 0}\cup\{\infty\}$, let $S^{(2, n)}(u, v, k)\subseteq S^{(2)}(u, v, k)$ denote the set of vertices $w$ lying on the segment $s_2(\gamma)$ of a geodesic $\gamma \in \Gamma(u, v, n)$ from $u$ to $v$, and satisfying $d_{H_{a, b}}(u, w) = k$. 
\end{df}

Note that $S^{(2)}(u, v, k) = \bigcup_{n\in\mathbb{Z}^{\geq0}\cup\{\infty\}}S^{(2, n)}(u, v, k)$, and $S^{(2, \infty)}(u, v, k) \subseteq \{(0, 0)\}$. 

\begin{lem} \label{s2n}
    For all vertices $u, v$ in $H_{a, b}$, integers $k$, and non-negative integers $n$, we have
    \begin{equation*}
        \left|S^{(2, n)}(u, v, k)\right|\leq 16p(p + 3b/a). 
    \end{equation*}
\end{lem}

\begin{proof}
    For any geodesic $\gamma$ in $H_{a, b}$ from $u$ to $v$ with $n(\gamma) = n$, $\gamma$ does not pass through any vertices in $(p^{n + 1}\mathbb{Z})^2$. Therefore, by Lemma~\ref{classification level n}, $\gamma$ is contained within an open $p^{n + 1}$-strip. Such an open $p^{n + 1}$-strip must contain the endpoints $u, v$ of $\gamma$. Therefore, there are at most two such open $p^{n + 1}$-strips: at most one horizontal open $p^{n + 1}$-strip and at most one vertical open $p^{n + 1}$-strip. 

    As the endpoints of $s_2(\gamma)$ lie in $(p^n\mathbb{Z})^2$, by Lemma~\ref{fractal inheritance lem 2}, $s_2(\gamma)$ consists of $p^n$-SPs. For an open $p^{n + 1}$-strip $R$ containing $u, v$, let $\Lambda (R)$ be the set of $p^n$-SPs $\lambda$ contained in $R$, each of which forms a segment of $s_2(\gamma)$ for some geodesic $\gamma\in \Gamma(u, v, n)$, and intersects $B_{H_{a, b}}(u, k)$. As in Lemma~\ref{s1}, for any open $p^{n + 1}$-strip $R$ containing $u, v$, each $\lambda\in \Lambda(R)$ contains at most one vertex in $S^{(2, n)}(u, v, k)$; and all vertices in $S^{(2, n)}(u, v, k)$ lie on some $\lambda\in \Lambda(R)$ for some open $p^{n + 1}$-strip $R$ containing $u, v$. Therefore, it suffices to show that there are at most $16p(p + 3b/a)$ such $p^n$-SPs $\lambda$. 

    Since there are at most two choices of the open $p^{n + 1}$-strip $R$ containing $u, v$, it suffices to show, for each such $R$, that $|\Lambda(R)|\leq 8p(p + 3b/a)$. Without loss of generality, assume that $R$ is horizontal. Let $u = (x, y)$. 

    Given any $\lambda\in \Lambda(R)$, let $w'$ be the endpoint of $\lambda$ that is closer to $u$ in $H_{a, b}$. Let $\gamma\in\Gamma(u, v, n)$ be a geodesic for which $\lambda$ is a segment of $s_2(\gamma)$. Let $\Tilde{u}(\gamma) = (\Tilde{x}, \Tilde{y})$ and $w' = (x', y')$. As $\Tilde{u}(\gamma)$ is a $p^n$-parent of $u$, $|x - \Tilde{x}| < p^n$. We note that all horizontal $p^n$-SPs contained in $R$ cost $p^na - a + b$, and the segment of $\gamma$ from $\Tilde{u}(\gamma)$ to $w'$ uses at least $\frac{1}{p^n}|x' - \Tilde{x}|$ such horizontal $p^n$-SPs. Therefore, 
    \begin{equation*}
        d_{H_{a, b}}(\Tilde{u}(\gamma), w') \geq (p^na - a + b)\frac{|x' - \Tilde{x}|}{p^n} > (p^na - a + b)\left(\frac{|x' - x|}{p^n} - 1\right). 
    \end{equation*}
    On the other hand, there is a path from $\Tilde{u}(\gamma)$ to $w'$ using exactly $\frac{1}{p^n}|x' - \Tilde{x}|$ horizontal $p^n$-SPs and at most $p - 2$ vertical $p^n$-SPs. As any $p^n$-SP costs at most $p^na - a + b$, we have
    \begin{equation*}
        d_{H_{a, b}}(\Tilde{u}(\gamma), w') \leq (p^na - a + b)\left(\frac{|x' - \Tilde{x}|}{p^n} + (p - 2)\right) < (p^na - a + b)\left(\frac{|x' - x|}{p^n} + p - 1\right). 
    \end{equation*}
    The distance $d_{H_{a, b}}(\Tilde{u}(\gamma), w')$ is related to $k$ via
    \begin{equation*}
        k - 3bp^n\leq k - d_{H_{a, b}}(u, \Tilde{u}(\gamma)) - d_{H_{a, b}}(w, w') = d_{H_{a, b}}(\Tilde{u}(\gamma), w')\leq k. 
    \end{equation*}
    Combining the three preceding inequalities, we have
    \begin{equation*}
        1 - p + \frac{k - 3bp^n}{p^na - a + b} < \frac{|x' - x|}{p^n} < 1 + \frac{k}{p^na - a + b}. 
    \end{equation*}
    As $x'\in p^n\mathbb{Z}$, there are at most $2\lceil p + \frac{3bp^n}{p^na - a + b}\rceil\leq 2(p + 3b/a + 1)$ choices of $x'$. As $y'\in p^n\mathbb{Z}$ is confined within an open $p^{n + 1}$-interval for $w' = (x', y')$ to lie in the open $p^{n + 1}$-strip $R$, there are at most $p - 1$ choices of $y'$. Hence, there are at most $2(p + 3b/a + 1)(p - 1)$ choices of $w'$. There are at most four choices of a $p^n$-SP $\lambda$ given one of its endpoints $w'$. Therefore, $|\Lambda(R)|\leq 8(p + 3b/a + 1)(p - 1)\leq 8p(p + 3b/a)$, as needed. 
\end{proof}

To uniformly upper bound the size of $S^{(2)}(u, v, k) = \bigcup_{n\in\mathbb{Z}^{\geq0}\cup\{\infty\}}S^{(2, n)}(u, v, k)$, it now remains to uniformly upper bound, across pairs of vertices $u, v$, the number of values of $n$ for which $S^{(2, n)}(u, v, k)$ is non-empty. 

\begin{df}
    Given vertices $u, v$ in $H_{a, b}$, let $N_H(u, v)$ (\emph{resp.}~$N_V(u, v)$) denote the set of non-negative integers $n$ for which there exists a geodesic $\gamma\in \Gamma(u, v, n)$ contained within an open horizontal (\emph{resp.}~vertical) $p^{n + 1}$-strip. 
\end{df}

By Lemma~\ref{classification level n}, $\Gamma(u, v, n)$ is non-empty if and only if $n\in N_H(u, v)\cup N_V(u, v)$. We shall now show the uniform boundedness of $|N_H(u, v)|$ across all vertices $u, v$ in $H_{a, b}$. 

\begin{lem} \label{n bound start}
    For vertices $u, v$ in $H_{a, b}$, if $n_1 < n_2$ are both elements of $N_H(u, v)$, then the difference $\Delta$ of the horizontal coordinates of $u, v$ satisfies
    \begin{equation*}
        a p^{n_1} + (b - a)\left(\frac{d_{\mathbb{Z}^2}(u, v)}{2p^{n_1}} - p\right) < d_{H_{a, b}}(u, v) - a\Delta < 8bp^{n_1 + 1} + \frac{b - a}{p^{n_1}}d_{\mathbb{Z}^2}(u, v).
    \end{equation*}
\end{lem}

\begin{proof}
    As $u, v$ are contained in the same open horizontal $p^{n_1 + 1}$-strip, their vertical coordinates differ by less than $p^{n_1 + 1}$. Hence, $d_{\mathbb{Z}^2}(u, v) < \Delta + p^{n_1 + 1}$. 

    We first prove the upper bound on $d_{H_{a, b}}(u, v) - a\Delta$. Take any $p^{n_1}$-parents $u', v'$ of $u, v$, respectively. By the triangle inequality, 
    \begin{equation*}
        d_{H_{a, b}}(u, v)\leq d_{H_{a, b}}(u, u') + d_{H_{a, b}}(u', v') + d_{H_{a, b}}(v', v). 
    \end{equation*}
    We now upper bound each of the three terms on the right-hand side. We have $d_{H_{a, b}}(u, u')\leq bd_{\mathbb{Z}^2}(u, u')\leq 2bp^{n_1}$, and similarly $d_{H_{a, b}}(v, v')\leq 2bp^{n_1}$. Lastly, by Lemma~\ref{fractal inheritance lem 2}, we have
    \begin{align*}
        d_{H_{a, b}}(u', v') &= d_{H_{p^{n_1}a, p^{n_1}a - a + b}}\left(\frac{1}{p^{n_1}}u', \frac{1}{p^{n_1}}v'\right)\\
        &\leq (p^{n_1}a - a + b)d_{\mathbb{Z}^2}\left(\frac{1}{p^{n_1}}u', \frac{1}{p^{n_1}}v'\right)\\
        &= \frac{p^{n_1}a - a + b}{p^{n_1}}d_{\mathbb{Z}^2}(u', v'). 
    \end{align*}
    Note that, again by the triangle inequality, $d_{\mathbb{Z}^2}(u', v')\leq d_{\mathbb{Z}^2}(u, v) + 4p^{n_1}$. Therefore, we have
    \begin{align*}
        d_{H_{a, b}}(u, v) - a\Delta&\leq 4bp^{n_1} + \frac{p^{n_1}a - a + b}{p^{n_1}}d_{\mathbb{Z}^2}(u, v) + 4(p^{n_1}a - a + b) - a\Delta\\
        &\leq 8bp^{n_1} + ad_{\mathbb{Z}^2}(u, v) + \frac{b - a}{p^{n_1}}d_{\mathbb{Z}^2}(u, v) - a\Delta\\
        &< 8bp^{n_1} + ap^{n_1 + 1} + \frac{b - a}{p^{n_1}}d_{\mathbb{Z}^2}(u, v)\\
        & < 8bp^{n_1 + 1} + \frac{b - a}{p^{n_1}}d_{\mathbb{Z}^2}(u, v), 
    \end{align*}
    establishing the upper bound on $d_{H_{a, b}}(u, v) - a\Delta$. We now turn to the lower bound. For $i = 1, 2$, let $\gamma_i\in \Gamma(u, v, n_i)$ which is contained in an open horizontal $p^{n_i + 1}$-strip $R_i$ containing $u, v$. Let $y_i$ be the vertical coordinate of $\Tilde{u}(\gamma_i)$. We have $y_1\in p^{n_1}\mathbb{Z}$ and $y_2\in p^{n_2}\mathbb{Z}\subseteq p^{n_1}\mathbb{Z}$. However, the entirety of $\gamma_1$, and in particular $\Tilde{u}(\gamma_1)$, is contained in the open horizontal $p^{n_1 + 1}$-strip. Therefore, $y_1\notin p^{n_1 + 1}\mathbb{Z}$ and hence $y_1\notin p^{n_2}\mathbb{Z}$. In particular, $y_1\neq y_2$. As they are both multiples of $p^{n_1}$, we have $|y_1 - y_2|\geq p^{n_1}$. Let $\delta$ be the closed walk formed by concatenating $\gamma_1$ with (the reversal of) $\gamma_2$. As $\delta$ passes through both $\Tilde{u}(\gamma_1), \Tilde{u}(\gamma_2)$, $\delta$ must pass through at least $2p^{n_1}$ vertical edges. Similarly, as $\delta$ passes through both $u$ and $v$, $\delta$ must pass through $2\Delta$ horizontal edges. Therefore, $\delta$ must pass through at least $2\Delta + 2p^{n_1}$ edges in total. 
    
    By Lemma~\ref{classification level n}, $s_2(\gamma_1)$ consists of $p^{n_1}$-SPs. The difference in the horizontal coordinates of $\Tilde{u}(\gamma_1)$ and $\Tilde{v}(\gamma_1)$ is at least $\Delta - 2p^{n_1}$. Therefore, $s_2(\gamma_1)$ passes through at least $\frac{\Delta}{p^{n_1}} - 2$ horizontal $p^{n_1}$-SPs. All horizontal $p^{n_1}$-SPs contained in the open horizontal $p^{n_1 + 1}$-strip $R_1$ contain a slow edge. Therefore, $s_2(\gamma_1)$, and hence $\gamma_1$, passes through at least $\frac{\Delta}{p^{n_1}} - 2$ slow edges. Therefore, $\delta$ must pass through at least $\frac{\Delta}{p^{n_1}} - 2$ slow edges. Since $\delta$ passes through at least $2\Delta + 2p^{n_1}$ edges, the cost of the closed walk $\delta$ (\emph{i.e.}~the sum of the weights of its edges) is at least $a(2\Delta + 2p^{n_1}) + (b - a)\left(\frac{\Delta}{p^{n_1}} - 2\right)$. On the other hand, the cost of $\delta$, being the sum of the costs of the geodesics $\gamma_1$, $\gamma_2$, is equal to $2d_{H_{a, b}}(u, v)$. Therefore, 
    \begin{equation*}
        d_{H_{a, b}}(u, v)\geq a(\Delta + p^{n_1}) + (b - a)\left(\frac{\Delta}{2p^{n_1}} - 1\right). 
    \end{equation*}
    As $d_{\mathbb{Z}^2}(u, v) < \Delta + p^{n_1 + 1}$, we have the lower bound 
    \begin{equation*}
        d_{H_{a, b}}(u, v) - a\Delta\geq a p^{n_1} + (b - a)\left(\frac{\Delta}{2p^{n_1}} - 1\right) > a p^{n_1} + (b - a)\left(\frac{d_{\mathbb{Z}^2}(u, v)}{2p^{n_1}} - p/2 - 1\right). \qedhere
    \end{equation*}
\end{proof}

\begin{lem} \label{n bound}
    For any vertices $u, v$ in $H_{a, b}$, we have
    \begin{equation*}
        |N_H(u, v)|\leq 2\log_p\left(\frac{18bp}{a}\right) + 1. 
    \end{equation*}
\end{lem}

\begin{proof}
    Let $T$ denote the set $N_H(u, v)$ with its maximum element removed (if one exists). As $|N_H(u, v)\backslash T|\leq 1$, it suffices to prove that $|T|\leq 2\log_p\left(\frac{18bp}{a}\right)$. For the sake of contradiction, we assume otherwise, that $|T| > 2\log_p\left(\frac{18bp}{a}\right)$. 
    
    By Lemma~\ref{n bound start}, for any $n\in T$, we have
    \begin{equation*}
        a p^n + (b - a)\left(\frac{d_{\mathbb{Z}^2}(u, v)}{2p^n} - p\right) < d_{H_{a, b}}(u, v) - a\Delta < 8bp^{n + 1} + \frac{b - a}{p^n}d_{\mathbb{Z}^2}(u, v), 
    \end{equation*}
    where $\Delta$ is the difference of the horizontal coordinates of $u, v$. Therefore, for any $n, m\in T$, we have
    \begin{equation*}
        a p^n + (b - a)\left(\frac{d_{\mathbb{Z}^2}(u, v)}{2p^n} - p\right) < 8bp^{m + 1} + \frac{b - a}{p^m}d_{\mathbb{Z}^2}(u, v). 
    \end{equation*}
    Hence, for any $n, m\in T$, we have
    \begin{equation*}
        a p^n + \frac{b - a}{2p^n}d_{\mathbb{Z}^2}(u, v) < 9bp^{m + 1} + \frac{b - a}{p^m}d_{\mathbb{Z}^2}(u, v)\leq \frac{9bp}{a}\left(a p^m + \frac{b - a}{2p^m}d_{\mathbb{Z}^2}(u, v)\right). 
    \end{equation*}
    In other words, for any $n, m\in T$, we have
    \begin{equation*}
        \frac{p^n}{c} + \frac{c}{p^n} < \frac{9bp}{a}\left(\frac{p^m}{c} + \frac{c}{p^m}\right), 
    \end{equation*}
    where $c = \sqrt{\frac{(b-a)d_{\mathbb{Z}^2}(u, v)}{2a}} > 0$. Since $|T| > 2\log_p\left(\frac{18bp}{a}\right)$, either $T\cap [\log_p c, \infty)$ or $T\cap [0, \log_p c]$ has at least $\log_p\left(\frac{18bp}{a}\right) + 1$ elements. If $|T\cap [\log_p c, \infty)|\geq \log_p\left(\frac{18bp}{a}\right) + 1$, let $n, m\in T\cap [\log_p c, \infty)$ with $n - m \geq \log_p\left(\frac{18bp}{a}\right)$. As $m \geq \log_p c$, we have $\frac{p^m}{c}\geq 1\geq \frac{c}{p^m}$. Therefore, 
    \begin{equation*}
        \frac{p^n}{c} + \frac{c}{p^n} > \frac{p^n}{c} \geq \frac{18bp}{a}\frac{p^m}{c} \geq \frac{18bp}{a}\frac{\frac{p^m}{c} + \frac{c}{p^m}}{2} = \frac{9bp}{a}\left(\frac{p^m}{c} + \frac{c}{p^m}\right), 
    \end{equation*}
    a contradiction. Similarly, if $|T\cap [0, \log_p c]|\geq \log_p\left(\frac{18bp}{a}\right) + 1$, let $n, m\in T\cap [0, \log_p c]$ with $m - n \geq \log_p\left(\frac{18bp}{a}\right)$. As $m \leq \log_p c$, we have $\frac{c}{p^m}\geq 1\geq \frac{p^m}{c}$. Therefore, 
    \begin{equation*}
        \frac{p^n}{c} + \frac{c}{p^n} > \frac{c}{p^n} \geq \frac{18bp}{a}\frac{c}{p^m} \geq \frac{18bp}{a}\frac{\frac{p^m}{c} + \frac{c}{p^m}}{2} = \frac{9bp}{a}\left(\frac{p^m}{c} + \frac{c}{p^m}\right), 
    \end{equation*}
    again, a contradiction. 
\end{proof}

We are now ready to show the uniform boundedness of $S^{(2)}(u, v, k)$ across all pairs of vertices $u, v$ in $H_{a, b}$ and integers $k$. 

\begin{lem} \label{s2}
    For all vertices $u, v$ in $H_{a, b}$ and integers $k$, we have
    \begin{equation*}
        \left|S^{(2)}(u, v, k)\right|\leq 64p(p + 3b/a)\log_p\left(\frac{18bp^2}{a}\right).
    \end{equation*}
\end{lem}

\begin{proof}
    For a non-negative integer $n\notin N_H(u, v)\cup N_V(u, v)$, $\Gamma(u, v, n) = \emptyset$, therefore, by definition, $S^{(2, n)}(u, v, k) = \emptyset$. Hence, $S^{(2)}(u, v, k) = \left(\bigcup_{n\in N_H(u, v)\cup N_V(u, v)}S^{(2, n)}(u, v, k)\right) \cup S^{(2, \infty)}(u, v, k)$. By Lemma~\ref{s2n}, $|S^{(2, n)}(u, v, k)|\leq 16p(p + 3b/a)$ for any non-negative integer $n$. By Lemma~\ref{n bound}, $|N_H(u, v)|\leq 2\log_p\left(\frac{18bp}{a}\right) + 1$, and analogously, $|N_V(u, v)|\leq 2\log_p\left(\frac{18bp}{a}\right) + 1$. Therefore, noting that $S^{(2, \infty)}(u, v, k) \subseteq \{(0, 0)\}$, we have
    \begin{equation*}
        |S^{(2)}(u, v, k)|\leq 16p(p + 3b/a)\cdot 2\left(2\log_p\left(\frac{18bp}{a}\right) + 1\right) + 1\leq 64p(p + 3b/a)\log_p\left(\frac{18bp^2}{a}\right). \qedhere
    \end{equation*}
\end{proof}

We are now ready to prove Proposition~\ref{main reduced prop}, which, along with Proposition~\ref{sec2 reduction prop}, completes the proof of Theorem~\ref{z2 gen ext thm}. 

\begin{proof}[Proof of Proposition~\ref{main reduced prop}]
    We shall prove that the weighted graph $H_{a, b}$ on $\mathbb{Z}^2$ with two positive integer weights $a, b$ has bounded geodesic slices for any choice of an odd integer $p\geq 3$ and positive integers $a < b$. For any pair $u, v$ of vertices in $H_{a, b}$, and any integer $k$, by Lemma~\ref{s1}, we have
    \begin{equation*}
        |S^{(1)}(u, v, k)|\leq 80\frac{b^2p^2}{a^2}, 
    \end{equation*}
    and by symmetry, 
    \begin{equation*}
        |S^{(3)}(u, v, k)|\leq 80\frac{b^2p^2}{a^2}. 
    \end{equation*}
    By Lemma~\ref{s2}, we have
    \begin{equation*}
        |S^{(2)}(u, v, k)|\leq 64p(p + 3b/a)\log_p\left(\frac{18bp^2}{a}\right). 
    \end{equation*}
    Therefore, the size of their union $S(u, v, k) = S^{(1)}(u, v, k)\cup S^{(2)}(u, v, k)\cup S^{(3)}(u, v, k)$ is bounded above by 
    \begin{equation*}
        |S(u, v, k)|\leq 160\frac{b^2p^2}{a^2} + 64p(p + 3b/a)\log_p\left(\frac{18bp^2}{a}\right). 
    \end{equation*}
    Hence $H_{a, b}$ has bounded geodesic slices. 
\end{proof}

\section{The hexagonal and triangular lattices} \label{other Euc}

In this section, we prove Theorem~\ref{general thm}, that all two-dimensional regular tilings are geodesically directable. Following the remarks in the introduction, it suffices to show that the hexagonal lattice $\{6, 3\}$ and the triangular lattice $\{3, 6\}$ are geodesically directable. 

Figure~\ref{fig:hex} shows a bounded extension $Z$ of the hexagonal lattice $\{6, 3\}$. The original lattice is drawn in black with additional edges of the bounded extension drawn in grey. Note that $Z$ is isomorphic to the square lattice $\mathbb{Z}^2 = \{4, 4\}$. By Theorem~\ref{z2 thm}, a further bounded extension $G$ of $Z$ has bounded geodesic slices. As $G$ is also a bounded extension of the hexagonal lattice $\{6, 3\}$, the hexagonal lattice is geodesically directable. 

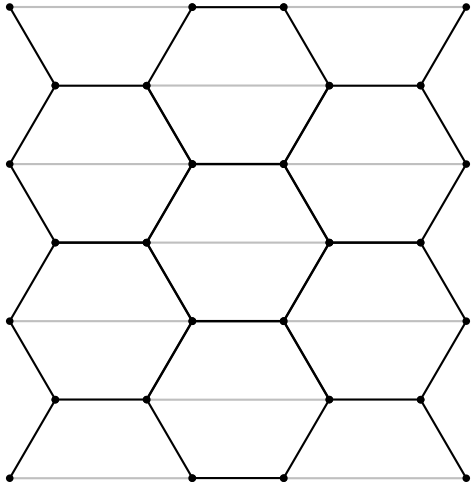
\begin{figure}
    \centering
    \begin{tikzpicture}[line join=round, line cap=round, thick, scale=1.2]
        
        \def\s{1.2}
        \def\u{0.835}
        \pgfmathsetmacro{\h}{sqrt(3)/2*\s}
        \pgfmathsetmacro{\cx}{1.5*\s}
        \pgfmathsetmacro{\cy}{sqrt(3)*\s}
        
        \tikzset{
          hex/.pic={
            \draw[gray!50] (\s,0) -- (-\s,0);
            \draw
              (\s,0) -- (\s/2,\h) -- (-\s/2,\h) -- (-\s,0) --
              (-\s/2,-\h) -- (\s/2,-\h) -- cycle;
            \foreach \P in {(\s,0),(\s/2,\h),(-\s/2,\h),(-\s,0),(-\s/2,-\h),(\s/2,-\h)}{
              \fill \P circle[radius=0.05];
            }
          }
        }
        
        \tikzset{
          hex1/.pic={
            \draw[gray!50] (\s,0) -- (-\s,0);
            \draw
                (\s/2,-\h) -- (\s,0);
            \foreach \P in {(\s/2,-\h), (\s,0), (-\s,0)}{
              \fill \P circle[radius=0.05];
            }
          }
        }
        
        \tikzset{
          hex2/.pic={
            \draw[gray!50] (\s,0) -- (-\s,0);
            \draw
                (-\s/2,\h) -- (-\s,0);
            \foreach \P in {(-\s/2,\h), (-\s,0), (\s,0)}{
              \fill \P circle[radius=0.05];
            }
          }
        }
        
        \tikzset{
          hex3/.pic={
            \draw[gray!50] (\s,0) -- (-\s,0);
            \draw
              (-\s,0) -- (-\s/2,-\h);
            \foreach \P in {(-\s,0), (-\s/2,-\h), (\s,0)}{
              \fill \P circle[radius=0.05];
            }
          }
        }
        
        \tikzset{
          hex4/.pic={
            \draw[gray!50] (\s,0) -- (-\s,0);
            \draw
              (\s,0) -- (\s/2,\h);
            \foreach \P in {(\s,0), (\s/2,\h), (-\s,0)}{
              \fill \P circle[radius=0.05];
            }
          }
        }
        
        \newcommand{\hexqr}[2]{
          \pgfmathsetmacro{\xx}{\cx*\u*(#1)}
          \pgfmathsetmacro{\yy}{\cy*\u*((#2)+0.5*(#1))}
          \path (\xx,\yy) pic {hex};
        }
        
        \newcommand{\hexqra}[2]{
          \pgfmathsetmacro{\xx}{\cx*\u*(#1)}
          \pgfmathsetmacro{\yy}{\cy*\u*((#2)+0.5*(#1))}
          \path (\xx,\yy) pic {hex1};
        }
        
        \newcommand{\hexqrb}[2]{
          \pgfmathsetmacro{\xx}{\cx*\u*(#1)}
          \pgfmathsetmacro{\yy}{\cy*\u*((#2)+0.5*(#1))}
          \path (\xx,\yy) pic {hex2};
        }
        
        \newcommand{\hexqrc}[2]{
          \pgfmathsetmacro{\xx}{\cx*\u*(#1)}
          \pgfmathsetmacro{\yy}{\cy*\u*((#2)+0.5*(#1))}
          \path (\xx,\yy) pic {hex3};
        }
        
        \newcommand{\hexqrd}[2]{
          \pgfmathsetmacro{\xx}{\cx*\u*(#1)}
          \pgfmathsetmacro{\yy}{\cy*\u*((#2)+0.5*(#1))}
          \path (\xx,\yy) pic {hex4};
        }
        
        \hexqr{0}{0}
        \hexqr{1}{0}
        \hexqr{-1}{0}
        \hexqr{0}{1}
        \hexqr{0}{-1}
        \hexqr{1}{-1}
        \hexqr{-1}{1}
        \hexqra{1}{1}
        \hexqrb{-1}{-1}
        \hexqrc{-1}{2}
        \hexqrd{1}{-2}    
    \end{tikzpicture}
    \caption{A bounded extension of the hexagonal lattice which is isomorphic to $\mathbb{Z}^2$}
    \label{fig:hex}
\end{figure}

Figure~\ref{fig:tri} shows that the triangular lattice $\{3, 6\}$ is isomorphic to a bounded extension of the square lattice $\mathbb{Z}^2 = \{4, 4\}$. Therefore, by Theorem~\ref{z2 gen ext thm}, the triangular lattice $\{3, 6\}$ is geodesically directable. 

\begin{figure}
    \centering
    \begin{tikzpicture}[line join=round, line cap=round, thick, scale=1.2]
        \def\s{0.8}
        \def\u{0.835}
        \tikzset{
          square/.pic={
            \draw
              (-\s,-\s) -- (\s,-\s) -- (\s,\s) -- (-\s,\s) -- cycle;
            \draw
              (-\s,-\s) -- (\s,\s);
            \foreach \P in {(-\s,-\s),(\s,-\s),(\s,\s),(-\s,\s)}{
                \fill \P circle[radius=0.05];
            }
          }
        }
        
        \newcommand{\sqij}[2]{
          \path (#1*2*\u*\s, #2*2*\u*\s) pic {square};
        }
        
        \foreach \i in {-1,0,1,2}{
          \foreach \j in {-1,0,1,2}{
            \sqij{\i}{\j}
          }
        }
    \end{tikzpicture}
    \caption{The triangular lattice as a bounded extension of $\mathbb{Z}^2$}
    \label{fig:tri}
\end{figure}
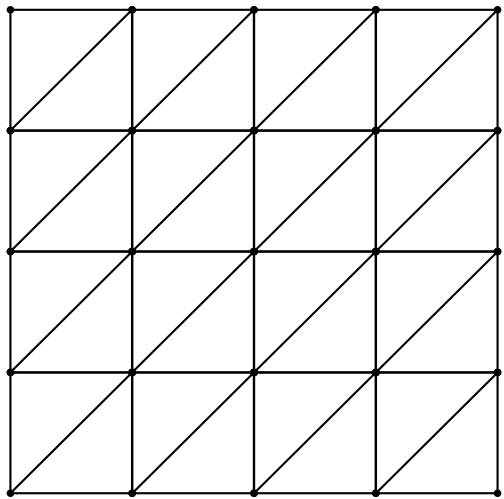

\section{Concluding Remarks}

The main object of this paper has been the construction of a bounded extension $G$ of $\mathbb{Z}^2$ with bounded geodesic slices (Theorem~\ref{z2 thm}). One could consider a strengthening of the bounded geodesic slices property, in which the number of geodesics in $G$ between any pair of endpoints is uniformly bounded. 

\begin{df}
    Let a graph $G$ have \emph{bounded geodesic multiplicity} if the number of geodesics between a pair of vertices $u, v$ is uniformly bounded. 
\end{df}

A natural question, therefore, is whether there exists a bounded extension $G$ of $\mathbb{Z}^2$ with bounded geodesic multiplicity. 

\begin{ques} \label{bounded geodesic multiplicity}
    Does there exist a bounded extension of $\mathbb{Z}^2$ with bounded geodesic multiplicity? 
\end{ques}

Note that in $\mathbb{Z}^2$ itself, the number of geodesics between two vertices $u, v$ can grow exponentially in the distance $d_{\mathbb{Z}^2}(u, v)$. For example, for any non-negative integer $x$, there are $\binom{2x}{x} = 2^{\Omega(x)}$ distinct geodesics between $(0, 0)$ and $(x, x)$. \emph{A priori}, therefore, one may expect a negative answer to Question~\ref{bounded geodesic multiplicity}. However, the affirmative resolution of the analogous question for bounded geodesic slices (Theorem~\ref{z2 thm}) shows that passing to a bounded extension may drastically alter the geodesic structure. 

\textbf{Acknowledgement.} The author is grateful to Oliver Breach for communicating this problem and its background. The author would also like to thank Timothy Gowers for helpful comments and suggestions. 

\bibliographystyle{abbrv}
\bibliography{mybib}

\begin{thebibliography}{10}

\bibitem{arutesuperconducting}
F.~Arute, K.~Arya, R.~Babbush, D.~Bacon, J.~C. Bardin, R.~Barends, R.~Biswas, S.~Boixo, F.~G. S.~L. Brandao, D.~A. Buell, B.~Burkett, Y.~Chen, Z.~Chen, B.~Chiaro, R.~Collins, W.~Courtney, A.~Dunsworth, E.~Farhi, B.~Foxen, A.~Fowler, C.~Gidney, M.~Giustina, R.~Graff, K.~Guerin, S.~Habegger, M.~P. Harrigan, M.~J. Hartmann, A.~Ho, M.~Hoffmann, T.~Huang, T.~S. Humble, S.~V. Isakov, E.~Jeffrey, Z.~Jiang, D.~Kafri, K.~Kechedzhi, J.~Kelly, P.~V. Klimov, S.~Knysh, A.~Korotkov, F.~Kostritsa, D.~Landhuis, M.~Lindmark, E.~Lucero, D.~Lyakh, S.~Mandr{\`a}, J.~R. McClean, M.~McEwen, A.~Megrant, X.~Mi, K.~Michielsen, M.~Mohseni, J.~Mutus, O.~Naaman, M.~Neeley, C.~Neill, M.~Y. Niu, E.~Ostby, A.~Petukhov, J.~C. Platt, C.~Quintana, E.~G. Rieffel, P.~Roushan, N.~C. Rubin, D.~Sank, K.~J. Satzinger, V.~Smelyanskiy, K.~J. Sung, M.~D. Trevithick, A.~Vainsencher, B.~Villalonga, T.~White, Z.~J. Yao, P.~Yeh, A.~Zalcman, H.~Neven, and J.~M. Martinis.
\newblock Quantum supremacy using a programmable superconducting processor.
\newblock {\em Nature}, 574(7779):505--510, 2019.

\bibitem{dualityreview}
B.~Bertini, P.~W. Claeys, and T.~Prosen.
\newblock Exactly solvable many-body dynamics from space-time duality.
\newblock {\em arXiv preprint arXiv:2505.11489}, 2025.

\bibitem{breach}
O.~Breach, B.~Placke, P.~W. Claeys, and S.~Parameswaran.
\newblock Solvable quantum circuits in tree+ 1 dimensions.
\newblock {\em arXiv preprint arXiv:2503.20927v2}, 2025.

\bibitem{Bridson-Haefliger}
M.~R. Bridson and A.~Haefliger.
\newblock {\em $\delta$-Hyperbolic Spaces and Area}, pages 398--437.
\newblock Springer Berlin Heidelberg, Berlin, Heidelberg, 1999.

\bibitem{coxeter_regular_1973}
H.~S.~M. Coxeter.
\newblock {\em Regular Polytopes}.
\newblock Dover Publications, New York, 3rd edition, 1973.

\bibitem{review}
M.~P. Fisher, V.~Khemani, A.~Nahum, and S.~Vijay.
\newblock Random quantum circuits.
\newblock {\em Annual Review of Condensed Matter Physics}, 14(Volume 14, 2023):335--379, 2023.

\bibitem{hastingsqLDPC}
M.~B. Hastings, J.~Haah, and R.~O'Donnell.
\newblock Fiber bundle codes: breaking the n 1/2 polylog (n) barrier for quantum ldpc codes.
\newblock In {\em Proceedings of the 53rd Annual ACM SIGACT Symposium on Theory of Computing}, pages 1276--1288, 2021.

\bibitem{kimsuperconducting}
Y.~Kim, A.~Eddins, S.~Anand, K.~X. Wei, E.~van~den Berg, S.~Rosenblatt, H.~Nayfeh, Y.~Wu, M.~Zaletel, K.~Temme, and A.~Kandala.
\newblock Evidence for the utility of quantum computing before fault tolerance.
\newblock {\em Nature}, 618(7965):500--505, 2023.

\bibitem{leverrierqLDPC}
A.~Leverrier, J.-P. Tillich, and G.~Z{\'e}mor.
\newblock Quantum expander codes.
\newblock In {\em 2015 IEEE 56th Annual Symposium on Foundations of Computer Science}, pages 810--824. IEEE, 2015.

\bibitem{pastawskiHaPPY}
F.~Pastawski, B.~Yoshida, D.~Harlow, and J.~Preskill.
\newblock Holographic quantum error-correcting codes: Toy models for the bulk/boundary correspondence.
\newblock {\em Journal of High Energy Physics}, 2015(6):1--55, 2015.

\end{thebibliography}

\appendix
\section{Explicit construction}

For the sake of practical implementations in quantum circuits, we include here an explicit bounded extension $G$ of $\mathbb{Z}^2$ with bounded geodesic slices. This construction results from applying a tightened version of the reduction argument in Section~\ref{intro and weight} to the weighted graph constructed in Section~\ref{construction sec} (with parameters $p = 3, a = 1, b = 2$). 

We form the graph $G$ from $\mathbb{Z}^2$ by adding vertical and horizontal edges of lengths two and four. For a non-zero integer $n$, let $\nu_3(n)$ denote the largest exponent $\alpha$ for which $3^\alpha \mid n$. For $x, y\in \mathbb{Z}$, if $x\neq 0$ and $3^{\nu_3(x)}\mid y - \frac{3^{\nu_3(x)} - 1}{2}$, we add the vertical edges $(4x, 4y)(4x, 4y + 2)$ and $(4x, 4y + 2)(4x, 4y + 4)$; otherwise, we add the vertical edge $(4x, 4y)(4x, 4y + 4)$. Similarly, if $y\neq 0$ and $3^{\nu_3(y)}\mid x - \frac{3^{\nu_3(y)} - 1}{2}$, we add the horizontal edges $(4x, 4y)(4x + 2, 4y)$ and $(4x + 2, 4y)(4x + 4, 4y)$; otherwise, we add the horizontal edge $(4x, 4y)(4x + 4, 4y)$. 

Figure~\ref{fig:def3} illustrates the portion of $G$ over $[0, 36]\times [0, 36]$. 

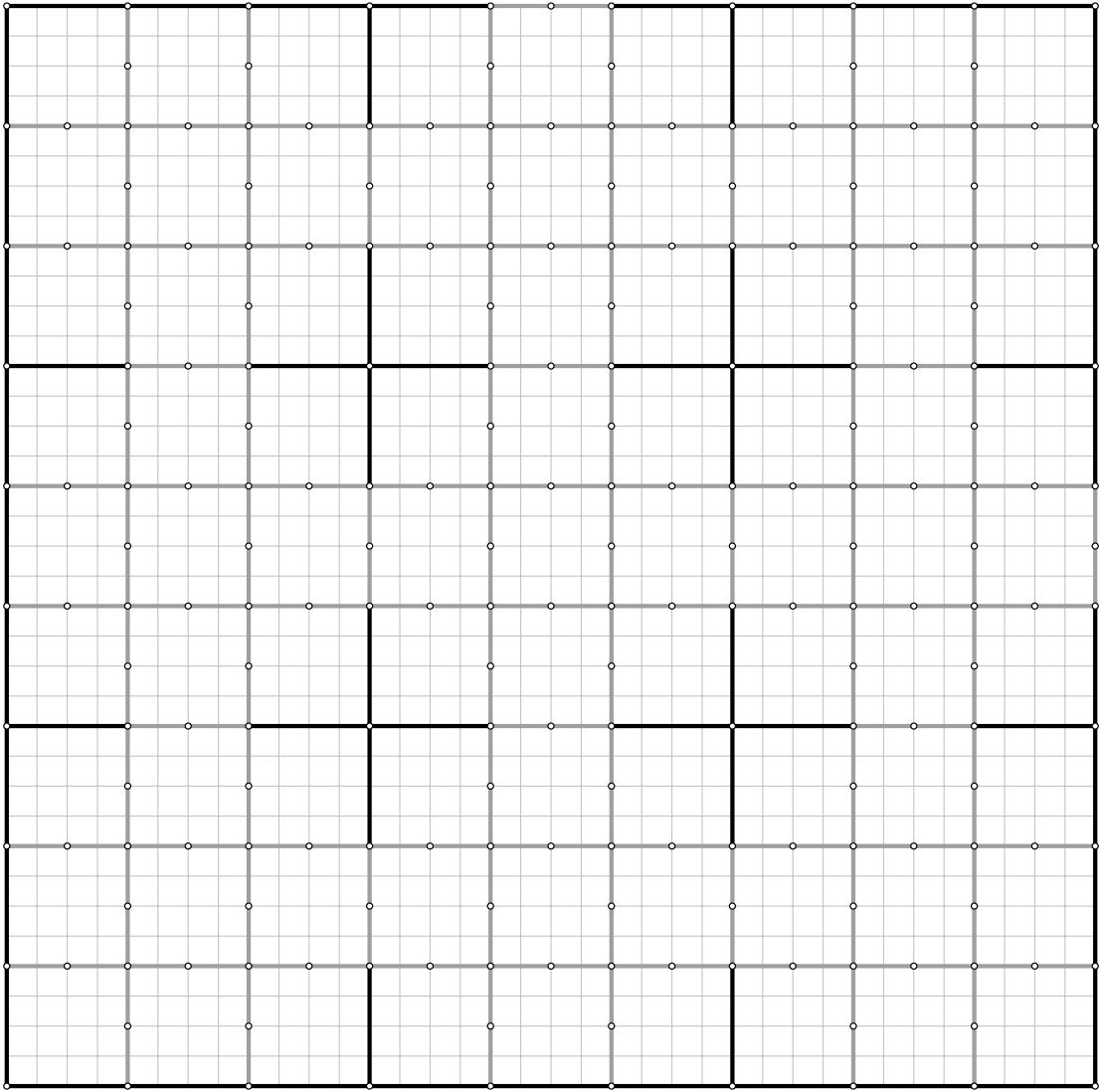
\begin{figure}
    \centering
    \begin{tikzpicture}[scale = 0.4]
        \draw[step=1cm,gray!50,very thin] (0,0) grid (36, 36);
        
        \foreach \x in {1, 2, 4, 5, 7, 8}{
            \foreach \y in {0, 1, 2, 3, 4, 5, 6, 7, 8}{
                \draw[gray!75, ultra thick] (4*\x, 4*\y) -- (4*\x, 4*\y + 2) -- (4*\x, 4*\y + 4);
                \fill (4*\x, 4*\y + 2) circle[radius=0.12];
                \fill[white] (4*\x, 4*\y + 2) circle[radius=0.08];
                \draw[gray!75, ultra thick] (4*\y, 4*\x) -- (4*\y + 2, 4*\x) -- (4*\y + 4, 4*\x);
                \fill (4*\y + 2, 4*\x) circle[radius=0.12];
                \fill[white] (4*\y + 2, 4*\x) circle[radius=0.08];
            }
        }

        \foreach \x in {0, 3, 6, 9}{
            \foreach \y in {0, 2, 3, 5, 6, 8}{
                \draw[ultra thick] (4*\x, 4*\y) -- (4*\x, 4*\y + 4);
                \draw[ultra thick] (4*\y, 4*\x) -- (4*\y + 4, 4*\x);
            }
        }

        \foreach \x in {3, 6}{
            \foreach \y in {1, 4, 7}{
                \draw[gray!75, ultra thick] (4*\x, 4*\y) -- (4*\x, 4*\y + 2) -- (4*\x, 4*\y + 4);
                \fill (4*\x, 4*\y + 2) circle[radius=0.12];
                \fill[white] (4*\x, 4*\y + 2) circle[radius=0.08];
                \draw[gray!75, ultra thick] (4*\y, 4*\x) -- (4*\y + 2, 4*\x) -- (4*\y + 4, 4*\x);
                \fill (4*\y + 2, 4*\x) circle[radius=0.12];
                \fill[white] (4*\y + 2, 4*\x) circle[radius=0.08];
            }
        }

        \foreach \x in {1, 4, 7}{
            \draw[ultra thick] (4*\x, 0) -- (4*\x + 4, 0);
            \draw[ultra thick] (0, 4*\x) -- (0, 4*\x + 4);
        }

        \foreach \x in {1, 7}{
            \draw[ultra thick] (4*\x, 36) -- (4*\x + 4, 36);
            \draw[ultra thick] (36, 4*\x) -- (36, 4*\x + 4);
        }

        \draw[gray!75, ultra thick] (16, 36) -- (18, 36) -- (20, 36);
        \fill (18, 36) circle[radius=0.12];   
        \fill[white] (18, 36) circle[radius=0.08];
        \draw[gray!75, ultra thick] (36, 16) -- (36, 18) -- (36, 20);
        \fill (36, 18) circle[radius=0.12];    
        \fill[white] (36, 18) circle[radius=0.08];

        \foreach \x in {0, 1, 2, 3, 4, 5, 6, 7, 8, 9}{
            \foreach \y in {0, 1, 2, 3, 4, 5, 6, 7, 8, 9}{
                \fill (4*\x, 4*\y) circle[radius=0.12];
                \fill[white] (4*\x, 4*\y) circle[radius=0.08];
            }
        }
    \end{tikzpicture}
    \caption{Illustration of the bounded extension $G$ of $\mathbb{Z}^2$, drawn over $[0, 36]\times [0, 36]$. Edges of the base lattice $\mathbb{Z}^2$ are drawn in light grey. New edges (edges in $G$ but not $\mathbb{Z}^2$, \emph{i.e.}~quantum gates) with length two are drawn in dark grey, and new edges with length four are drawn in black. Endpoints of new edges are marked.}
    \label{fig:def3}
\end{figure}
\end{document}